\documentclass[]{sig-alternate-05-2015}
 \newtheorem{theorem}{Theorem}[section]
 \newtheorem{proposition}[theorem]{Proposition}
 \newtheorem{corollary}[theorem]{Corollary}

\usepackage{bm}
\usepackage{bbm} 
\usepackage{framed}
\usepackage{algorithm}
\usepackage{algorithmic}
\usepackage{overpic}
\newtheorem{lemma}{Lemma}
\newtheorem{definition}{Definition}

\begin{document}
\title{Throughput-Optimal Broadcast in Wireless Networks with Dynamic Topology} 
\numberofauthors{3}
    \author{
      \alignauthor Abhishek Sinha\\ 
      \affaddr{Laboratory for Information and Decision Systems}\\
      \affaddr{MIT}     
      \email{sinhaa@mit.edu}
  \and \and     \alignauthor Leandros Tassiulas\\
      \affaddr{Electrical Engg. and Yale Institute of Network Science\\ }
      \affaddr{Yale University} \\    
   \email{leandros.tassiulas@yale.edu}
 \and \and 
     \alignauthor Eytan Modiano\\  
     \affaddr{Laboratory for Information and Decision Systems}\\
      \affaddr{MIT}     
      \email{modiano@mit.edu}  
%
          }
\maketitle

%
\begin{abstract}
We consider the problem of throughput-optimal broadcasting in time-varying wireless networks, whose underlying topology is restricted to Directed Acyclic Graphs (DAG). Previous broadcast algorithms route packets along spanning trees. In large networks with time-varying connectivities, these trees are difficult to compute and maintain. In this paper we propose a new online throughput-optimal broadcast algorithm which makes packet-by-packet scheduling and routing decisions, obviating the need for maintaining any global topological structures, such as spanning-trees. Our algorithm relies on system-state information for making transmission decisions and hence, may be thought of as a generalization of the well-known \emph{back-pressure algorithm} which makes point-to-point unicast transmission decisions based on queue-length information, without requiring knowledge of end-to-end paths. Technically, the back-pressure algorithm is derived by stochastically stabilizing the network-queues. However, because of packet-duplications associated with broadcast, the work-conservation principle is violated and \emph{queuing processes} are difficult to define in the broadcast problem. To address this fundamental issue, we identify certain state-variables which behave like \emph{virtual queues} in the broadcast setting. 
By stochastically stabilizing these virtual queues, we devise a throughput-optimal broadcast policy. We also derive new characterizations of the broadcast-capacity of time-varying wireless DAGs and derive an efficient algorithm to compute the capacity exactly under certain assumptions, and a poly-time approximation algorithm for computing the capacity under less restrictive assumptions.
\end{abstract}
\section{Introduction} \label{intro-section}
The problem of efficiently disseminating packets, arriving at a source node, to a subset of nodes in a network, is known as the \emph{Multicast problem}. In the special case when the packets are to be distributed among all nodes, the corresponding problem is referred to as the \emph{Broadcast problem}.  Multicasting and broadcasting is considered to be a fundamental network functionality, which enjoys numerous practical applications ranging from military communications \cite{milcom_app}, disaster management using mobile adhoc networks (MANET) \cite{ge2004overlay}, to streaming services for live web television \cite{ip_tv_multicast} etc. \\
There exists a substantial body of literature addressing different aspects of this problem in various networking settings. An extensive survey of various multicast routing protocols for MANET is provided in \cite{multicast_survey}. The authors of \cite{gandhi2003minimizing} consider the problem of minimum latency broadcast of a finite set of messages in MANET. This problem is shown to be NP-hard. To address this issue, several approximation algorithms are proposed in \cite{huang2007minimum}, all of which rely on construction of certain network-wide broadcast-trees. Cross-layer solutions for multi-hop multicasting in wireless network are given in \cite{yuan2006cross} and \cite{Ho2005}. These algorithms involve network coding, which introduces additional complexity and exacerbates end-to-end delay. The authors of \cite{swati} propose a multicast scheduling and routing protocol which balances load among a set of pre-computed spanning trees, which are challenging to compute and maintain in a scalable fashion. The authors of \cite{towsley2008rate} propose a local control algorithm for broadcasting in a wireless network for the so called \emph{scheduling-free model}, in which an oracle is assumed to make interference-free scheduling decisions. This assumption, as noted by the authors themselves, is not practically viable.\\
In this paper we build upon the recent work of \cite{sinha_DAG} and consider the problem of throughput-optimal broadcasting in a wireless network with time-varying connectivity. Throughout the paper, the overall network-topology will be restricted to a directed acyclic graph (DAG). We first characterize the broadcast-capacity of time-varying wireless networks and propose an exact and an approximation algorithm to compute it efficiently. Then we propose a dynamic link-activation and packet-scheduling algorithm that, unlike any previous algorithms, obviates the need to maintain any global topological structures, such as spanning trees, yet achieves the capacity. In addition to throughput-optimality, the proposed algorithm enjoys the attractive property of \emph{in-order} packet-delivery, which makes it particularly useful in various online applications, e.g. VoIP and live multimedia communication \cite{chu2001enabling}. Our algorithm is model-oblivious in the sense that its operation does not rely on detailed statistics of the random arrival or network-connectivity processes. We also show that the throughput-optimality of our algorithm is retained when the control decisions are made using \emph{locally} available and possibly imperfect, state information. \\
Notwithstanding the vast literature on the general topic of broadcasting, to the best of our knowledge, this is the first work addressing throughput-optimal broadcasting in time-varying wireless networks with store and forward routing. Our main technical contributions are the following: 
\begin{itemize}
\item We define the broadcast-capacity for wireless networks with time-varying connectivity and characterize it mathematically and algorithmically. We show that broadcast-capacity of time-varying wireless directed acyclic networks can be computed efficiently under some assumptions. We also derive a tight-bound for the capacity for a general setting and utilize it to derive an efficient approximation algorithm to compute it.
\item We propose a throughput-optimal dynamic routing and scheduling algorithm for broadcasting in a wireless DAGs with time-varying connectivity. This algorithm is of \emph{Max-Weight} type and uses the idea of \emph{in-order} delivery to simplify its operation. To the best of our knowledge, this is the first throughput-optimal dynamic algorithm proposed for the broadcast problem in wireless networks. 
\item We extend our algorithm to the setting when the nodes have access to infrequent state updates. We show that the throughput-optimality of our algorithm is preserved even when the rate of inter-node communication is made arbitrarily small.
\item We illustrate our theoretical findings through illustrative numerical simulations.
\end{itemize}
The rest of the paper is organized as follows. Section \ref{sys_model} introduces the wireless network model. Section \ref{capacity_section} defines and characterizes the broadcast capacity of a wireless DAG. It also provides an exact and an approximation algorithm to compute the broadcast-capacity. Section \ref{optimal_policy_section} describes our capacity-achieving broadcast algorithm for DAG networks. Section \ref{intermittent_connectivity} extends the algorithm to the setting of broadcasting with imperfect state information. Section \ref{simulation_section} provides numerical simulation results to illustrate our theoretical findings. Finally, in section \ref{conclusion} we summarize our results and conclude the paper. 
\section{Network Model} \label{sys_model}
First we describe the basic wireless network model without time-variation. Subsequently, we will incorporate time-variation in the basic model. A static wireless network is modeled by a directed graph $\mathcal{G}=(V,E,\bm{c},\mathcal{M})$, where $V$ is the set of nodes, $E$ is the set of \emph{directed} point-to-point links\footnote{We assume all transmit and receiving antennas to be directed and hence all transmissions to be point-to-point \cite{beygelzimer2008benefits}.}, the vector $\bm{c} = (c_{ij})$ denotes  capacities of the edges when the corresponding links are activated and $\mathcal{M}\subset \{0,1\}^{|E|}$ is the set of incidence-vectors corresponding to all feasible link-activations complying with the interference-constraints. The structure of the activation-set $\mathcal{M}$ depends on the interference model, e.g., under the primary or node-exclusive interference model \cite{shroff-tutorial}, $\mathcal{M}$ corresponds to the set of all \emph{matchings} on the graph $\mathcal{G}$.  There are a total of $|V|=n$ nodes and $|E|=m$ edges in the network. Time is slotted and at time-slot $t$, any subset of links complying with the underlying interference-constraint may be activated. At most $c_{ij}$ packets can be transmitted in a slot from node $i$ to node $j$, when link $(i,j)$ is activated. \\ 
Let $\texttt{r}\in V$ be the \emph{source} node. At slot $t$, $A(t)$ packets arrive at the source. The arrivals, $A(t)$, are i.i.d. over slots with mean $\mathbb{E}(A(t))=\lambda$.
Our problem is to efficiently disseminate the packets to all nodes in the network.
\subsection{Notations and Nomenclature:} All vectors in this paper are assumed to be column vectors. For any set $\mathcal{X} \subset \mathbb{R}^k$, its convex-hull is denoted by $\mathrm{conv}(\mathcal{X})$. Let $\big(U, V\setminus U\big)$ be a disjoint partition of the set of vertices $V$ of the graph $\mathcal{G}$, such that the source $\texttt{r} \in U$ and $U \subsetneq V$. Such a partition is called a \emph{proper-partition}. To each proper partition corresponding to the set $U$, associate the \emph{proper-cut} vector $\bm{u} \in \mathbb{R}^m$, defined as follows:
\begin{eqnarray} \label{cut_vector}
 \bm{u}_{i,j}&=&c_{i,j} \hspace{10pt} \mathrm{if}\hspace{3pt} i\in U, j \in V\setminus U\\
 &=& 0 \hspace{17pt} \mathrm{otherwise}  \nonumber 
\end{eqnarray}
Denote the special, single-node proper-cuts by $U_j\equiv V\setminus \{j\}$, and the corresponding cut-vectors by $\bm{u}_j$, $\forall j \in V\setminus \{\texttt{r}\}$. The set of all proper-cut vectors in the graph $\mathcal{G}$ is denoted by $\mathcal{U}$. \\
The \emph{in-neighbours} of a node $j$ is defined as the set of all nodes $i\in V$ such that there is a directed edge $(i,j)\in E$. It is denoted by the set $\partial^{\mathrm{in}}(j)$, i.e.,
\begin{eqnarray}
\partial^{\mathrm{in}}(j)= \{i \in V: (i,j)\in E\}
\end{eqnarray}
Similarly, we define the \emph{out-neighbours} of a node $j$ as follows 
\begin{eqnarray}
\partial^{\mathrm{out}}(j) =\big\{i \in V: (j,i)\in E\big\}
\end{eqnarray} 
For any two vectors  $\bm{x}$ and $\bm{y}$ in $\mathbb{R}^m$, define the component-wise product $\bm{z}\equiv \bm{x} \odot \bm{y}$ to be a vector in $\mathbb{R}^m$ such that $z_i= x_iy_i, 1\leq i \leq m$. 

For any set $\mathcal{S}\subset \mathbb{R}^m$ and any vector $\bm{v} \in \mathbb{R}^m$, $\bm{v} \odot \mathcal{S}$, denotes the set of vectors obtained as the component-wise product of the vector $\bm{v}$ and the elements of the set $\mathcal{S}$, i.e.,  
\begin{eqnarray}
\bm{v}\odot \mathcal{S}=\big\{ \bm{y}\in \mathbb{R}^m: \bm{y}=\bm{v}\odot \bm{s}, \bm{s}\in \mathcal{S}\big\}
\end{eqnarray} 

Also, the usual dot product between two vectors $\bm{x}, \bm{y} \in \mathbb{R}^m$ is defined as, 
\begin{eqnarray*}
	\bm{x}\cdot \bm{y} = \sum_{i=1}^{m} x_iy_i
\end{eqnarray*}

\subsection{Model of Time-varying Wireless Connectivity} \label{model}
Now we incorporate time-variation into our basic framework described above. In a wireless network, the channel-SINRs vary with time because of random fading, shadowing and mobility \cite{tse2005fundamentals}. To model this, we consider a simple ON-OFF model where an individual link can be in one of the two states, namely \textsc{ON} and \textsc{OFF}. In an OFF state, the capacity of a link is zero \footnote{Generalization of the ON-OFF model, to multi-level discretization of link-capacity is straight-forward.}. Thus at a given time, the network can be in any one configuration, out of the set of all possible network configurations $\Xi$. Each element $\sigma \in \Xi$ corresponds to a sub-graph $\mathcal{G}(V,E_\sigma) \subset \mathcal{G}(V,E) $, with $E_\sigma \subset E$, denoting the set of links that are ON. At a given time-slot $t$, one of the configuration $\sigma(t) \in \Xi$  is realized. The configuration at time $t$ is represented by the vector $\bm{\sigma}(t) \in \{0,1\}^{|E|}$, where
\begin{eqnarray*}
\bm{\sigma}(e,t)=\begin{cases}
	1,  \hspace{5pt}\text{if}\hspace{5pt} e \in E_{\bm{\sigma}(t)}\\
0, \hspace{10pt}\text{otherwise}.
\end{cases}
\end{eqnarray*}
At a given time-slot $t$, the network controller may activate a set of non-interfering links that are ON.  \\
 The network-configuration process $\{\bm{\sigma}(t)\}_{t\geq 1}$ evolves in discrete-time according to a stationary ergodic process with the stationary distribution $\{p(\sigma)\}_{\sigma \in \Xi}$ \cite{Kamthe:2013:IWL:2555947.2529991}, where
\begin{eqnarray}
\sum_{\sigma \in \Xi}p(\sigma)=1, \hspace{5pt} p(\sigma)>0, \hspace{5pt}\forall \sigma \in \Xi
\end{eqnarray} 

 Since the underlying physical processes responsible for time-variation are often spatially-correlated \cite{agrawal2009correlated}, \cite{patwari2008effects}, the  distribution of the link-states is assumed to follow an arbitrary joint-distribution. The detailed parameters of this process depend on the ambient physical environment, which is often difficult to measure. In particular, it is unrealistic to assume that the broadcast-algorithm has knowledge of the parameters of the process $\bm{\sigma}(t)$. Fortunately, our proposed dynamic throughput-optimal broadcast algorithm does not require the statistical characterization of the configuration-process $\bm{\sigma}(t)$ or  its stationary-distribution $p(\bm{\sigma})$. This makes our algorithm robust and suitable for use in time-varying wireless networks. 
\section{Definition and Characterization of Broadcast Capacity} \label{capacity_section}
Intuitively, a network supports a broadcast rate $\lambda$ if there exists a scheduling policy under which all network nodes receive distinct packets at  rate $\lambda$. The broadcast-capacity of a network is the maximally supportable broadcast rate by any policy. 
Formally, we consider a class $\Pi$ of scheduling policies where each policy $\pi\in\Pi$ consists of a sequence of actions $\{\pi_t\}_{t\geq 1}$, executed at every slot $t$. Each action $\pi_{t}$ consists of two operations: 
\begin{itemize}
\item The scheduler observes the current network-configuration $\sigma(t)$ and activates a subset of links by choosing a feasible activation vector  $\boldsymbol s(t)\in\mathcal{M}_{\sigma(t)}$. Here $\mathcal{M}_\sigma$ denotes the set of all feasible link-activation vectors in the sub-graph $\mathcal{G}(V,E_\sigma)$, complying with the underlying interference constraints. As an example, under the primary interference constraint, $\mathcal{M}_\sigma$ is given by the set of all \emph{matchings} \cite{diestel2005graph} of the sub-graph $\mathcal{G}(V,E_\sigma)$.\\
Analytically, elements from the set $\mathcal{M}_\sigma$ will be denoted by their corresponding $|E|$-dimensional binary incidence-vectors, whose component corresponding to edge $e$ is identically zero if $e \notin E_\sigma$.
\item Each node $i$ forwards a subset of packets (possibly empty) to node $j$ over an activated link $(i, j) \in \bm{\sigma}(t)$, subject to the link capacity constraint. The class $\Pi$ includes policies that may use all past and future information, and may forward any subset of packets over a link, subject to the link-capacity constraint. 
\end{itemize}
  
To formally introduce the notion of broadcast capacity, we define the random variable $R_i^{\pi}(T)$ to be the number of distinct packets received by node $i \in V$ up to time $T$, under a policy $\pi\in \Pi$. The time average $\liminf_{T\to \infty} R^{\pi}_i(T)/T$ is  the rate of packet-reception at  node $i$.
 
\begin{definition}
A policy $\pi \in \Pi$ is called a 
{``broadcast policy of rate $\lambda$''} 
if 
all nodes receive distinct packets at rate $\lambda$, i.e.,
\begin{eqnarray} \label{bcdef}
\min_{i\in V} \liminf_{T\to \infty} \frac{1}{T} R^{\pi}_i(T)= \lambda, \hspace{10pt} \mathrm{w.p.}\hspace{3pt} 1
\end{eqnarray}
where $\lambda$ is the packet arrival rate at the source node $\texttt{r}$.
\end{definition}
\begin{definition} \label{capacity_def}
The broadcast capacity $\lambda^*$ of a network is defined to be the supremum of all arrival rates $\lambda$, for which there exists a broadcast policy $\pi\in\Pi$ of rate $\lambda$.
\end{definition}
In the following subsection, we derive an upper-bound on broadcast-capacity, which immediately follows from the previous definition. 

\subsection{An Upper-bound on Broadcast Capacity} \label{broadcast_ub_proof}

 Consider a policy $\pi \in \Pi$ that achieves a broadcast rate of at least  $\lambda^* -\epsilon$, for an $\epsilon >0$. Such a policy $\pi$ exists due to the definition of the broadcast capacity $\lambda^{*}$ in Definition \ref{capacity_def}.

Now consider any proper-cut $U$ of the network $\mathcal{G}$. By definition of a proper-cut, there exists a node $i \notin U$. Let $\bm{s}^{\pi}(t, \bm{\sigma}(t)) = (s_{e}^{\pi}(t), e\in E)$ be the link-activation vector chosen by policy $\pi$ in slot $t$, upon observing the current-configuration $\bm{\sigma}(t)$. The maximum number of packets that can be transmitted across the cut $U$ in slot $t$ is upper-bounded by the total capacity of all activated links across the cut-set $U$, which is given by $\sum_{e\in E_{U}} c_{e} s_{e}^{\pi}(t,\bm{\sigma}(t))$. Hence, the number of distinct packets received by node $i$ by time $T$ is upper-bounded by the total available capacity across the cut $U$ up to time $T$, subject to link-activation decisions of the policy $\pi$. In other words, we have

\begin{eqnarray} \label{bound_packet}
R_i^{\pi}(T) \leq \sum_{t=1}^{T} \sum_{e\in E_{U}} c_{e} s_{e}^{\pi}(t, \bm{\sigma}(t)) 
= \bm{u}\cdot \sum_{t=1}^{T} \bm{s}^{\pi}(t, \bm{\sigma}(t))
\end{eqnarray}
i.e., 
\begin{eqnarray*}
\frac{R_i^{\pi}(T)}{T} \leq \bm{u}\cdot\bigg(\frac{1}{T} \sum_{t=1}^{T}\bm{s}^{\pi}(t,\bm{\sigma}(t))\bigg),
\end{eqnarray*}
where the cut-vector $\bm{u}\in \mathbb{R}^m$, corresponds to the cut-set $U$, as in Eqn.\eqref{cut_vector}. It follows that,
\begin{eqnarray} 
\lambda^*-\epsilon  &\stackrel{(a)}{\leq}& \min_{j\in V} \liminf_{T\to \infty} \frac{R_j^{\pi}(T)}{T} \nonumber \leq  \liminf_{T\to \infty} \frac{R_i^{\pi}(T)}{T} \\
&\leq&  \liminf_{T\to \infty}\bm{u}\cdot\bigg(\frac{1}{T} \sum_{t=1}^{T}\bm{s}^{\pi}(t,\bm{\sigma}(t))\bigg), \label{bound2}
\end{eqnarray}
where (a) follows from the fact that $\pi$ is a broadcast policy of rate at least $\lambda^*-\epsilon$. Since the above inequality holds for all proper-cuts $\bm{u}$, we have 
\begin{eqnarray} \label{bound3}
\lambda^*-\epsilon \leq \min_{\bm{u} \in \mathcal{U}}  \liminf_{T\to \infty}\bm{u}\cdot\bigg(\frac{1}{T} \sum_{t=1}^{T}\bm{s}^{\pi}(t,\bm{\sigma}(t))\bigg)
\end{eqnarray}

The following technical lemma will prove to be useful for deriving an upper-bound on the broadcast-capacity.
\begin{framed}
\begin{lemma} \label{conv_hull}
 For any policy $\pi \in \Pi$, and any proper-cut vector $\bm{u}$,  there exists a collection of vectors $\big(\bm{\beta}_{\sigma}^{\pi} \in \mathrm{conv}(\mathcal{M}_\sigma)\big)_{\sigma\in \Xi}$, such that, the following holds w.p. $1$
\begin{eqnarray*}
\min_{\bm{u} \in \mathcal{U}}\liminf_{T\to \infty}  \bm{u}\cdot\bigg(\frac{1}{T} \sum_{t=1}^{T}\bm{s}^{\pi}(t,\bm{\sigma}(t))\bigg) \\
= \min_{\bm{u}\in \mathcal{U}}\bm{u}\cdot \bigg(\sum_{\sigma \in \Xi}p(\sigma)\bm{\beta}_\sigma^{\pi}\bigg)
\end{eqnarray*}
\end{lemma}
\end{framed}

%
%
%
%
%
%
The above lemma essentially replaces the minimum cut-set bound of an arbitrary activations in \eqref{bound3}, by the minimum cut-set bound of a stationary randomized activation, which is easier to handle.
Combining Lemma~\ref{conv_hull} with Eqn.~\eqref{bound3}, we conclude that for the policy $\pi \in \Pi$, there exists a collection of vectors $\{\bm{\beta}^\pi_{\sigma} \in \mathrm{conv}(\mathcal{M}_\sigma)\}_{\sigma \in \Xi}$ such that
\begin{equation} \label{bdcutC}
\lambda^*-\epsilon \leq \min_{\bm{u} \in \mathcal{U}} \bm{u}\cdot\bigg( \sum_{\sigma \in \Xi} p(\bm{\sigma})\bm{\beta}^{\pi}_{\sigma} \bigg)
\end{equation} 

Maximizing the RHS of Eqn. \eqref{bdcutC} over all vectors $\big\{\bm{\beta}_{\sigma} \in \mathrm{conv}(\mathcal{M}_\sigma),$ $\sigma \in \Xi\big\}$ and letting $\epsilon \searrow 0$, we have the following universal upper-bound on the broadcast capacity $\lambda^*$
\begin{eqnarray} \label{bc_ob}
\lambda^* \leq \max_{\bm{\beta}_{\sigma} \in \text{conv}(\mathcal{M}_\sigma)}\min_{\bm{u} \in \mathcal{U}} \bm{u}\cdot\bigg( \sum_{\sigma \in \Xi} p(\bm{\sigma})\bm{\beta}_{\sigma} \bigg)
\end{eqnarray}
Specializing the above bound for single-node cuts of the form $\bm{U}_j=(V\setminus \{j\}) \to \{j\}, \forall j \in V \setminus \{\texttt{r}\}$, we have the following upper-bound 
\begin{eqnarray}\label{cap_bd}
\lambda^* \leq \max_{\bm{\beta}_{\sigma} \in \text{conv}(\mathcal{M}_\sigma)}\min_{j \in V\setminus \{\texttt{r}\}} \bm{u}_j\cdot\bigg( \sum_{\sigma \in \Xi} p(\bm{\sigma})\bm{\beta}_{\sigma} \bigg)
\end{eqnarray}
It will be shown in Section \ref{optimal_policy_section} that in a DAG, our throughput-optimal policy $\pi^*$ achieves a broadcast-rate equal to the RHS of the bound \eqref{cap_bd}. Thus we have the following theorem 
\begin{framed}
\begin{theorem} \label{cap_th}
The broadcast-capacity $\lambda^*_{\mathrm{DAG}}$ of a time-varying wireless DAG 
is given by: 
\begin{eqnarray} \label{capacity_expr}
\lambda^*_{\mathrm{DAG}} = \max_{\bm{\beta}_{\sigma} \in \emph{conv}(\mathcal{M}_\sigma), \sigma \in \Xi }\min_{j \in V\setminus \{\texttt{r}\}} \bm{u}_j\cdot\bigg( \sum_{\sigma \in \Xi} p(\bm{\sigma})\bm{\beta}_{\sigma} \bigg)
\end{eqnarray}
\end{theorem}
\end{framed}
The above theorem shows that for computing the broadcast-capacity of a wireless DAG, taking minimum over the single-node cut-sets $\{u_j, j\in V \setminus\{\texttt{r}\}\}$ suffice (c.f. Eqn. \eqref{bc_ob}).
\subsection{An Illustrative Example of Capacity Computation} \label{example_comp}
In this section, we work out a simple example to illustrate the previous results.
\begin{figure} [h] 
\centering
\begin{minipage}{\textwidth}
\begin{overpic}[width=0.22\textwidth]{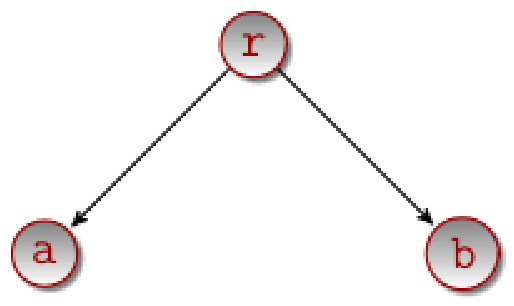}
\put(18,-5){Wireless network}
\end{overpic}
\begin{overpic}[width=0.22\textwidth]{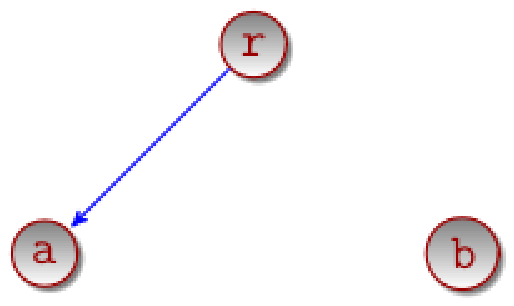}
\put(18,-5){Configuration $\sigma_1$}
\end{overpic}
\end{minipage}
\end{figure}
\begin{figure} [h] 
\begin{minipage}{\textwidth}
\begin{overpic}[width=0.22\textwidth]{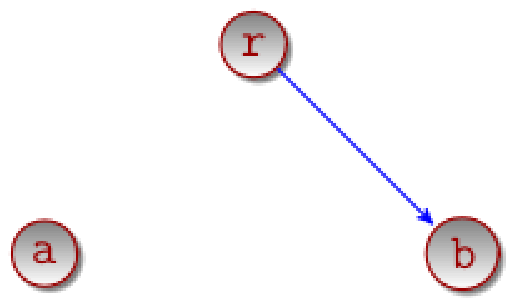}
\put(18,-5){Configuration $\sigma_2$}
\end{overpic}
\begin{overpic}[width=0.22\textwidth]{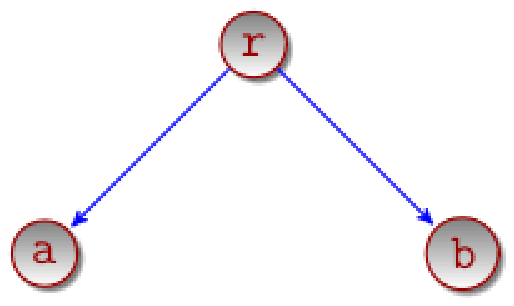}
\put(18,-5){Configuration $\sigma_3$}
\end{overpic}
\end{minipage}
\end{figure}

\begin{figure}[!h]
\center
\begin{overpic}[width=0.22\textwidth]{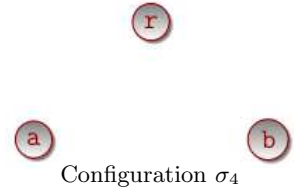}
\end{overpic}
\put(-90,-5){Configuration $\sigma_4$}
\caption{A Wireless Network and its four possible configurations}
\label{net}
\end{figure}
Consider the simple wireless network shown in Figure \eqref{net}, with node $\texttt{r}$ being the source. The possible network configurations $\sigma_i, i=1,2,3,4$ are also shown. One packet can be transmitted over a link if it is ON. Moreover, since the links are assumed to be point-to-point, even if both the links $\texttt{ra}$ and $\texttt{rb}$ are ON at a slot $t$ (i.e., $\sigma(t)=\sigma_3$), a packet can be transmitted over one of the links only. Hence, the sets of feasible activations are given as follows:
\begin{eqnarray*}
\mathcal{M}_{\sigma_1}=\{\begin{pmatrix}1\\0 \end{pmatrix}\}, \mathcal{M}_{\sigma_2}=\{\begin{pmatrix} 0 \\ 1 \end{pmatrix}\},\\
 \mathcal{M}_{\sigma_3}=\{\begin{pmatrix} 1 \\0 \end{pmatrix}, \begin{pmatrix}0 \\ 1 \end{pmatrix}\},\mathcal{M}_{\sigma_4}=\phi. 	
\end{eqnarray*}
Here the first coordinate corresponds to activating the edge $\texttt{ra}$ and the second coordinate corresponds to activating the edge $\texttt{rb}$. 
\\
 To illustrate the effect of link-correlations on broadcast-capacity, we consider three different joint-distributions $p(\bm{\sigma})$, all of them having the following marginal 
 
 \begin{eqnarray*}
 p(\texttt{ra}=\mathrm{ON})=p(\texttt{ra}=\mathrm{OFF})=\frac{1}{2} \\
 p(\texttt{rb}=\mathrm{ON})=p(\texttt{rb}=\mathrm{OFF})=\frac{1}{2}	
 \end{eqnarray*}

\paragraph{Case 1: Zero correlations} In this case, the links $\texttt{ra}$ and $\texttt{rb}$ are ON w.p. $\frac{1}{2}$ independently at every slot, i.e., 
\begin{eqnarray}
	p(\sigma_i)=1/4, \hspace{10pt} i=1,2,3,4 
\end{eqnarray}

  It can be easily seen that the broadcast capacity, as given in Eqn. \eqref{capacity_expr}, is achieved when in configurations $\sigma_1$ and $\sigma_2$, the edges $\texttt{ra}$ and $\texttt{rb}$ are activated w.p. $1$ respectively and in the configuration $\sigma_3$ the edges $\texttt{ra}$ and $\texttt{rb}$ are activated with probability $\frac{1}{2}$ and $\frac{1}{2}$. In other words, an optimal activation schedule of a corresponding stationary randomized policy is given as follows:
\begin{eqnarray*}
\bm{\beta}_{\sigma_1}^*=\big( 1 \hspace{15pt} 0\big)', \bm{\beta}_{\sigma_2}^*=\big( 0 \hspace{15pt} 1\big)', \bm{\beta}_{\sigma_3}^*=\big( \frac{1}{2} \hspace{15pt} \frac{1}{2}\big)'
\end{eqnarray*}
The optimal broadcast capacity can be computed from Eqn. \eqref{capacity_expr} to be $\lambda^*=\frac{1}{4}+0+\frac{1}{4}\times \frac{1}{2}=\frac{3}{8}$. \\
\paragraph{Case 2: Positive correlations} In this case, assume that the edges $\texttt{ra}$ and $\texttt{rb}$ are positively correlated, i.e.,  we have 
\begin{eqnarray*}
	p(\bm{\sigma}_1)=p(\bm{\sigma}_2)=0; \hspace{3pt}
	p(\bm{\sigma}_3)=p(\bm{\sigma}_4)=\frac{1}{2}
\end{eqnarray*}
Then it is clear that half of the slots are wasted when both the links are OFF (i.e., in the configuration $\bm{\sigma}_4$). When the network is in configuration $\bm{\sigma}_3$, an optimal randomized activation is to choose one of the two links uniformly at random and send packets over it. Thus 
\begin{eqnarray*}
	\bm{\beta}_{\sigma_3}^*=\big( \frac{1}{2} \hspace{15pt} \frac{1}{2}\big)'
\end{eqnarray*}
The optimal broadcast-capacity, computed from Eqn. \eqref{capacity_expr} is $\lambda^*=\frac{1}{4}$. 
\paragraph{Case 3: Negative correlations} In this case, we assume that the edges $\texttt{ra}$ and $\texttt{rb}$ are negatively correlated, i.e., we have 
\begin{eqnarray*}
	p(\bm{\sigma}_1)=p(\bm{\sigma}_2)=\frac{1}{2}; \hspace{3pt} p(\bm{\sigma}_3)=p(\bm{\sigma}_4)=0
\end{eqnarray*}
It is easy to see that in this case, a capacity-achieving activation strategy is to send packets over the link whichever is ON. The broadcast-capacity in this case is $\lambda^*=\frac{1}{2}$, the highest among the above three cases. \\
In this example, with an arbitrary joint distribution of network-configurations $\{p(\sigma_i), i=1,2,3,4\}$, it is a matter of simple calculation to obtain the optimal activations $\bm{\beta}^*_{\bm{\sigma}_i}$ in Eqn. \eqref{capacity_expr}. However it is clear that for an arbitrary network with arbitrary activations $\mathcal{M}$ and configuration sets $\Xi$, evaluating \eqref{capacity_expr} is non-trivial. In the following section we study this problem under some simplifying assumptions.  
 
\subsection{Efficient Computation of Broadcast Capacity}
In this section we study the problem of \emph{efficient computation} of the Broadcast Capacity $\lambda^*$ of a wireless DAG, given by Eqn. \eqref{capacity_expr}. In particular, we show that when the number of possible network configurations $|\Xi|(n)$ grows polynomially with $n$ (the number of nodes in the network), there exists a strongly polynomial-time algorithm to compute $\lambda^*$ under the primary-interference constraint. Polynomially-bounded network-configurations arise, for example, when the set $\Xi(n)$ consists of all subgraphs of the graph $\mathcal{G}$ with at most $d$ number of edges,  for some fixed integer $d$. In this case $|\Xi(n)|$ can be bounded as follows 
\begin{eqnarray*}
|\Xi|(n) \leq \sum_{k=0}^{d}\binom{m}{k} = \mathcal{O}(n^{2d}),	
\end{eqnarray*}
where $m ( =\mathcal{O}(n^2)$) is the number of edges in the graph $\mathcal{G}$.

\begin{framed}
\begin{theorem}[Efficient Computation of $\lambda^*$] \label{algo_comp}
Suppose that there exists a polynomial $q(n)$ such that, for a wireless $\mathrm{DAG}$ network $\mathcal{G}$ with $n$ nodes, the number of possible network configurations $|\Xi|(n)$ is bounded polynomially in $n$, i.e., $|\Xi|(n)=\mathcal{O}(q(n))$. Then, there exists a strongly $\mathsf{poly-time}$ algorithm to compute the broadcast-capacity of the network under the primary interference constraints.
\end{theorem}
\end{framed}
Although only polynomially many network configurations are allowed, we emphasize that Theorem \eqref{algo_comp} is highly non-trivial. This is because, each network-configuration $\sigma \in \Xi$ itself contains exponentially many possible activations (matchings). The key combinatorial result that leads to Theorem \eqref{algo_comp} is the existence of an efficient separator oracle for the matching-polytope for any arbitrary graph \cite{schrijver2003combinatorial}. We first reduce the problem of broadcast-capacity computation of a DAG to an LP with exponentially many constraints. Then invoking the above separator oracle, we show that this LP can be solved in strongly polynomial-time. 
\begin{proof}
See Appendix \ref{Algo_proof}.
\end{proof}
\subsection{Simple Bounds on $\lambda^*$}
Using Theorem \eqref{algo_comp} we can, in principle, compute the broadcast-capacity $\lambda^*$ of a wireless DAG with polynomially many network configurations. However, the complexity of the exact computation of $\lambda^*$ grows substantially with the number of the possible configurations $|\Xi|(n)$. Moreover, Theorem \eqref{algo_comp} does not apply when $|\Xi|(n)$ can no longer be bounded by a polynomial in $n$. A simple example of exponentially large $|\Xi|(n)$ is when a link $e$ is ON w.p. $p_e$ independently at every slot, for all $e \in E$.\\
 To address this issue, we obtain bounds on $\lambda^*$, whose computational complexity is independent of the size of $|\Xi|$. These bounds are conveniently expressed in terms of the broadcast-capacity of the static network $\mathcal{G}(V,E)$ without time-variation, i.e. when $|\Xi|=1$ and $E_\sigma =E, \sigma \in \Xi $. Let us denote the broadcast-capacity of the static network by $\lambda^*_{\text{stat}}$. Specializing Eqn. \eqref{capacity_expr} to this case, we obtain 
\begin{eqnarray} \label{whole_cap}
	\lambda^*_{\text{stat}} = \max_{\bm{\beta} \in \text{conv}(\mathcal{M})} \min_{j \in V \setminus \{\texttt{r}\}} \bm{u}_j\cdot \bm{\beta}.
\end{eqnarray}
Using Theorem \eqref{algo_comp}, $\lambda^*_\text{stat}$ can be computed in poly-time under the primary-interference constraint.\\ 
Now consider an arbitrary joint distribution $p(\bm{\sigma})$ such that each link is ON uniformly with probability $p$, i.e., 
\begin{eqnarray}\label{p_act}
	\sum_{\bm{\sigma} \in \Xi : \bm{\sigma}(e)=1}p(\bm{\sigma})=p, \hspace{10pt} \forall e \in E.
\end{eqnarray} 
 We have the following bounds: 
\begin{framed}
\begin{lemma}[Bounds on Broadcast Capacity]\label{bd_lemma}
	\begin{eqnarray*}
		p  \lambda^*_{\emph{stat}}\leq \lambda^* \leq \lambda^*_{\emph{stat}}.
	\end{eqnarray*}
\end{lemma}
\end{framed}
\begin{proof}
	See Appendix \ref{bd_proof}.
\end{proof}
Generalization of the above Lemma to the setting, where the links are ON with non-uniform probabilities, may also be obtained in a similar fashion. \\
Note that, in our example \ref{example_comp} the bounds in Lemma \ref{bd_lemma} are tight. In particular, here the value of the parameter $p=\frac{1}{2}$, the lower-bound is attained in case (2) and the upper-bound is attained in case (3). \\
 The above lemma immediately leads to the following corollary:
\begin{framed}
\begin{corollary}\label{approx-comp}
\emph{(}\textsc{Approximation-algorithm for computing $\lambda^*$}\emph{)}. Assume that, under the stationary distribution $p(\bm{\sigma})$, probability that any link is ON is $p$, uniformly for all links. Then, there exists a poly-time $p$-approximation algorithm to compute the broadcast-capacity $\lambda^*$ of a $\mathrm{DAG}$, under the primary-interference constraints. 
\end{corollary}
\end{framed}
\begin{proof}
	See Appendix \ref{approx-comp-proof}. 
\end{proof}
In the following section, we are concerned with  designing a dynamic and throughput-optimal broadcast policy for a time-varying wireless DAG network. 

\section{Throughput-Optimal Broadcast Policy for Wireless DAGs} \label{optimal_policy_section}
The classical approach of solving the throughput-optimal broadcast problem in the case of a static, wired network is to compute a set of edge-disjoint spanning trees of maximum cardinality (by invoking Edmonds' tree-packing theorem \cite{edmonds}) and then routing the incoming packets to all nodes via these pre-computed trees \cite{swati}.
In the time-varying wireless setting that we consider here, because of frequent and random changes in topology, routing packets over a fixed set of spanning trees is no-longer optimal. In particular, part of the network might become disconnected from time-to-time, and it is not clear how to select an optimal set of trees to disseminate packets. The problem becomes even more complicated when the underlying statistical model of the network-connectivity process (in particular, the stationary distribution $\{p(\bm{\sigma}), \bm{\sigma} \in \Xi\}$) is unknown, which is often the case in mobile adhoc networks. Furthermore, wireless interference constraints add another layer of complexity, rendering the optimal dynamic broadcasting problem in wireless networks extremely challenging. \\
In this section we propose an online, dynamic, throughput-optimal broadcast policy for time-varying wireless DAG networks, that does not need to compute or maintain any global topological structures, such as spanning trees. Interestingly, we show that the broadcast-algorithm that was proposed in \cite{sinha_DAG} for static wireless network, generalizes well to the time-varying case. As in \cite{sinha_DAG}, our algorithm also enjoys the attractive feature of \emph{in-order} packet delivery. The key difference between the algorithm in \cite{sinha_DAG} and our dynamic algorithm is in link-scheduling. In particular, in our algorithm, the activation sets are chosen based on current network-configuration $\bm{\sigma}(t)$. \\
\subsection{Throughput-Optimal Broadcast Policy $\pi^*$}
All policies $\pi \in \Pi$, that we consider in this paper, comprise of the following two sub-modules which are executed at every time-slot $t$: 
\begin{itemize} 
\item  $\pi(\mathcal{A})$ (\textbf{Activation-module}): activates a subset of links, subject to the interference constraint and the current network-configuration $\bm{\sigma}(t)$.
\item $\pi(\mathcal{S})$ (\textbf{Packet-Scheduling module}): schedules a subset of packets over the activated links.
\end{itemize}
Following the treatment in \cite{sinha_DAG}, we first restrict our attention to a sub-space $\Pi^{\mathrm{in-order}}$, in which the broadcast-algorithm is required to follow the so-called \emph{in-order} delivery property, defined as follows 
\begin{definition}[Policy-space $\Pi^{\mathrm{in-order}}$ \cite{sinha_DAG}]
A policy $\pi$ belongs to the space $\Pi^{\mathrm{in-order}}$ if all incoming packets are serially indexed as $\{1,2,3,\ldots \}$ according to their order of arrival at the source $\texttt{r}$ and  a node can receive a packet $p$ at time $t$, if and only if it has received the packets $\{1,2,,\ldots, p-1\}$ by the time $t$. 
\end{definition} 
As a consequence of the \emph{in-order} delivery, the state of received packets in the network  at time-slot $t$ may be succinctly represented by the $n$-dimensional vector $\bm{R}(t)$, where $R_i(t)$ denotes the index of the \emph{latest} packet received by node $i$ by time $t$. We emphasize that this succinct network-state representation by the vector $\bm{R}(t)$ is valid only under the action of policies in the space $\Pi^{\mathrm{in-order}}$. This compact representation of the packet-state results in substantial simplification of the overall state-space description. This is because, to completely specify the current packet-configurations in the network in the general policy-space $\Pi$, we need to specify the identity of each individual packets that are received by different nodes. \\
To exploit the special structure that a directed acyclic graph offers, it would be useful to constrain the packet-scheduler $\pi(\mathcal{S})$ further to the following policy-space $\Pi^* \subset \Pi^{\mathrm{in-order}}$.
\begin{definition}[Policy-space $\Pi^* \subset \Pi^{\mathrm{in-order}}$ \cite{sinha_DAG}]
A broadcast policy $\pi$ belongs to the space $\Pi^*$ if $\pi \in \Pi^{\mathrm{in-order}}$ and $\pi$  satisfies the additional constraint that a packet $p$ can be received by a node $j$ at time $t$ if all in-neighbours of the node $j$ have received  the packet $p$ by the time $t$. 
\end{definition}
The above definition is further illustrated in Figure \ref{pi*_figure}. The variables $X_j(t)$ and $i_t^*(j)$ appearing in the Figure are defined subsequently in Eqn. \eqref{var_def}.\\
\begin{figure} [h!] 
\centering
\begin{overpic}[width=0.39\textwidth]{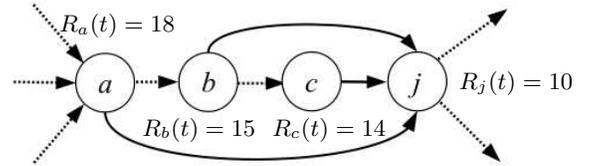}
  \put(11,28){\footnotesize $R_a(t)=18$}
  \put(27,8){\footnotesize $R_b(t)=15$}
  \put(52,8){\footnotesize $R_c(t)=14$}
  \put(88,17){\footnotesize $R_j(t)=10$}
  \end{overpic}
  \caption{\small Under a policy $\pi\in \Pi^*$, the set of packets available for transmission to node $j$ at slot $t$ is $\{11,12,13,14\}$,
  which are available at all in-neighbors of node $j$.
  The in-neighbor of $j$ inducing the smallest packet deficit is $i^*_t(j)=c$, and $X_{j}(t) = 4$. }
  \label{pi*_figure}
\end{figure}
It is easy to see that for all policies $\pi \in \Pi^*$, the packet scheduler $\pi(\mathcal{S})$ is \emph{completely} specified. Hence, to specify a policy in the space $\Pi^*$, we need to define the activation-module $\pi(\mathcal{A})$ only. \\
Towards this end, let $\mu_{ij}(t)$ denote the rate (in packets per slot) allocated to the edge $(i,j)$ in the slot $t$ by a policy $\pi \in \Pi^*$, for all $(i,j)\in E$. Note that, the allocated rate $\bm{\mu}(t)$ is constrained by the current network configuration $\bm{\sigma}(t)$ at slot $t$. In other words, we have 
\begin{eqnarray}
\bm{\mu}(t) \in \bm{c}\odot \mathcal{M}_{\bm{\sigma}(t)}, \hspace{5pt} \forall t
\end{eqnarray}
This implies that, under any randomized activation
\begin{eqnarray}
\mathbb{E}{\bm{\mu}(t)} \in \bm{c}\odot \mathrm{conv}( \mathcal{M}_{\bm{\sigma}(t)}), \hspace{5pt} \forall t
\end{eqnarray}
In the following lemma, we show that for all policies $\pi \in \Pi^*$,  certain state-variables $\bm{X}(t)$, derived from the state-vector $\bm{R}(t)$, satisfy so-called \emph{Lindley recursion} \cite{lindley1952theory} of queuing theory. Hence these variables may be thought of as \emph{virtual queues}. This technical result will play a central role in deriving a \emph{Max-Weight} type throughput-optimal policy $\pi^*$, which is obtained by stochastically stabilizing these virtual-queues. \\
For each $j \in V \setminus \{\texttt{r}\}$, define 
\begin{eqnarray} 
X_j(t) = \min_{i \in \partial^{\mathrm{in}}(j)} \big(R_i(t)-R_j(t)\big) \label{X-def}\label{var_def}\\
i_t^*(j) =\arg \min_{i \in \partial^{\mathrm{in}}(j)} \big(R_i(t)-R_j(t)\big), \label{i_t} 
\end{eqnarray}

where in Eqn. \eqref{i_t}, ties are broken lexicographically. The variable $X_j(t)$ denotes the minimum packet deficit of node $j$ with respect to any of its in-neighbours. Hence, from the definition of the policy-space $\Pi^*$, it is clear that $X_j(t)$ is the maximum number of packets that a node $j$ can receive from its in-neighbours at time $t$, under any policy in $\Pi^*$.\\
The following lemma proves a ``queuing-dynamics" of the variables $X_j(t)$, under any policy $\pi \in \Pi^*$. 
\begin{framed}
\begin{lemma}[\cite{sinha_DAG}] \label{lindley_lemma}
Under all policies in $\pi \in \Pi^*$, we have 
\begin{eqnarray}\label{dyn}
X_j(t+1) \leq \bigg(X_j(t)- \sum_{k \in \partial^{\mathrm{in}}(j)}\mu_{kj}(t)\bigg)^+ \nonumber \\
+ \sum_{m \in \partial^{\mathrm{in}}(i_t^*(j))} \mu_{mi_t^*(j)}(t) 
\end{eqnarray}
\end{lemma}
\end{framed}
Lemma \eqref{lindley_lemma} shows that the variables $\big(X_j(t),j\in V\setminus \{\texttt{r}\}\big)$ satisfy Lindley recursions in the policy-space $\Pi^*$. Interestingly, unlike the corresponding unicast problem \cite{tassiulas}, there is no ``physical queue" in the system.\\
 Continuing correspondence with the unicast problem, the next lemma shows that any activation module $\pi(\mathcal{A})$ that  ``stabilizes" the \emph{virtual queues} $\bm{X}(t)$ for all arrival rates $\lambda<\lambda^*$, constitutes a throughput optimal broadcast-policy for a wireless DAG network. 
\begin{framed}
\begin{lemma} \label{stability_lemma}
Suppose that, the underlying topology of the wireless network is a $\mathrm{DAG}$. If under the action of a broadcast policy $\pi \in \Pi^*$, for all arrival rates $\lambda< \lambda^*$, the virtual queue process $\{\bm{X}(t)\}_{0}^{\infty}$ is rate-stable, i.e., \
\begin{eqnarray*}
\limsup_{T\to \infty} \frac{1}{T}\sum_{j\neq \texttt{r}}X_j(T)=0, \hspace{5pt}\mathrm{w.p.}\hspace{2pt} 1,
\end{eqnarray*}
then $\pi$ is a throughput-optimal broadcast policy for the $\mathrm{DAG}$ network.  
\end{lemma}
\end{framed}
\begin{proof}
See Appendix \eqref{stability_lemma_proof}. 
\end{proof}
Equipped with Lemma \eqref{stability_lemma}, we now set out to derive a dynamic activation-module $\pi^*(\mathcal{A})$ to stabilize the virtual-queue process $\{\bm{X}(t)\}_{0}^{\infty}$ for all arrival rates $\lambda < \lambda^*$. Formally, the structure of the module $\pi^*(\mathcal{A})$ is given by a mapping of the following form: 
\begin{eqnarray*}
\pi^*(\mathcal{A}) : (\bm{X}(t), \bm{\sigma}(t)) \to \mathcal{M}_{\bm{\sigma}(t)}
\end{eqnarray*}
Thus, the module ${\pi}^*(\mathcal{A})$ is stationary and dynamic as it depends on the current value of the state-variables and the network-configuration only. This activation-module is different from the policy described in \cite{sinha_DAG} as the latter is meant for static wireless networks and hence, does not take into account the time-variation of network configurations, which is the focus of this paper. \\
To describe $\pi^*(\mathcal{A})$, we first define the node-set 
\begin{eqnarray}
K_j(t)=\{m \in \partial^{\mathrm{out}}(j) : j= i_t^*(m)\}
\end{eqnarray}
where the variables $i_t^*(m)$ are defined earlier in Eqn. \eqref{i_t}. 
The activation-module $\pi^*(\mathcal{A})$ is described in Algorithm 1. The resulting policy in the space $\Pi^*$ with the activation-module $\pi^*(\mathcal{A})$ is called $\pi^*$. \\
\begin{algorithm} 
\caption{A Throughput-optimal Activation Module $\pi^*(\mathcal{A})$}
\begin{algorithmic}[1]
\STATE To each link $(i,j)\in E$, assign a weight as follows:
\begin{eqnarray} \label{weight_comp}
W_{ij}(t)=
\begin{cases} X_j(t)-\sum_{k\in K_j(t)}X_k(t) , \hspace{2pt} \mathrm{if} \hspace{2pt} \bm{\sigma}_{(i,j)}(t)=1\\
0, \hspace{15pt} \mathrm{o.w.} 
\end{cases} 
\end{eqnarray}
\STATE Select an activation $s^*(t)\in \mathcal{M}_{\bm{\sigma}(t)}$ as follows: 
\begin{eqnarray}
s^*(t) \in \arg \max_{\bm{s} \in \mathcal{M}_{\bm{\sigma}(t)}} \bm{s} \cdot \big(\bm{c} \odot \bm{W}(t)\big)
\end{eqnarray}
\STATE Allocate rates on the links as follows: 
\begin{eqnarray}
\bm{\mu}^*(t)= \bm{c}\odot \bm{s}^*(t)
\end{eqnarray}
\end{algorithmic}
\end{algorithm}
Note that, in steps (1) and (2) above, the computation of link-weights and link-activations depend explicitly on the current network-configuration $\bm{\sigma}(t)$. As anticipated, in the following lemma, we show that the activation-module $\pi^*(\mathcal{A})$ stochastically stabilizes the virtual-queue process $\{\bm{X}(t)\}_{0}^{\infty}$. 
\begin{framed}
\begin{lemma} \label{stabilizing_X_lemma}
For all arrival rates $\lambda < \lambda^*$, under the action of the policy $\pi^*$ in a DAG, the virtual-queue process $\{\bm{X}(t)\}_{0}^{\infty}$ is rate-stable, i.e.,
\begin{eqnarray*}
\limsup_{T\to \infty} \frac{1}{T}\sum_{j\neq \texttt{r}}X_j(T)=0, \hspace{5pt}\mathrm{w.p.}\hspace{2pt} 1
\end{eqnarray*}
\end{lemma}
\end{framed}
The proof of this lemma is centered around a Lyapunov-drift argument \cite{neely2010stochastic}. Its complete proof is provided in Appendix \eqref{stabilizing_X_lemma_proof}.\\
Combining the lemmas \eqref{stability_lemma} and \eqref{stabilizing_X_lemma}, we immediately obtain the main result of this section
\begin{framed}
\begin{theorem} \label{pi_star_optimality}
The policy $\pi^*$ is a throughput-optimal broadcast policy in a time-varying wireless $\mathrm{DAG}$ network.
\end{theorem}
\end{framed}
\section{Throughput-Optimal Broadcasting with Infrequent Inter-node Communication} \label{intermittent_connectivity}
In practical mobile wireless networks, it is unrealistic to assume knowledge of network-wide packet-state information by every node at every slot. This is especially true in the case of time-varying wireless networks, where network-connectivity changes frequently. In this section we extend the main results of section \ref{optimal_policy_section} by considering the setting where the nodes make control decisions with \emph{imperfect} packet-state information that they currently possess. We will show that the dynamic broadcast-policy $\pi^*$ retains its throughput-optimality even in this challenging scenario.  

\paragraph{State-Update Model} We assume that two nodes $i$ and $j$ can mutually update their knowledge of the set of packets received by the other node, only at those slots with positive probability, when the corresponding wireless-link $(i,j)$ is in ON state. Otherwise, it continues working with the outdated packet state-information.  Throughout this section, we assume that the nodes have perfect information about the current network-configuration $\bm{\sigma}(t)$. \\
Suppose that, the latest time prior to time $t$ when packet-state update was made across the link $(i,j)$ is $t-T_{(i,j)}(t)$. Here $T_{(i,j)}(t)$ is a random variable, supported on the set of non-negative integers. Assume for simplicity, that the network configuration process $\{\bm{\sigma}(t)\}_{0}^{\infty}$ evolves according to a finite-state, positive recurrent Markov-Chain, with the stationary distribution $\{p(\bm{\sigma})>0, \bm{\sigma}\in \Xi\}$. With this assumption, $T_{(i,j)}(t)$ is related to the first-passage time in the finite-state positive recurrent chain $\{\bm{\sigma}(t)\}_{0}^{\infty}$. Using standard theory \cite{gallager2012discrete}, it can be shown that the random variable $T(t)\equiv \sum_{(i,j) \in E} T_{(i,j)}(t)$ has bounded expectation for all time $t$ . \\
\paragraph{Analysis of $\pi^*$ with Imperfect Packet-State Information}
Consider running the policy $\pi^*$, where each node $j$ now computes the weights $W'_{ij}(t)$, given by Eqn.\eqref{weight_comp}, of the in-coming links $(i,j) \in E$, based on the latest packet-state information available to it. In particular, for each of its in-neighbour $i \in \partial^{\mathrm{in}}(j)$, the node $j$ possess the following information of the number of packets received by node $i$:
\begin{eqnarray} \label{delayed_update}
R_i'(t)=R_i(t-T_{(ij)}(t))
\end{eqnarray}
Now, if the packet-scheduler module $\pi'(\mathcal{S})$ of a broadcast-policy $\pi'$ takes scheduling decision based on the imperfect state-information $\bm{R}'(t)$ (instead of the true state $\bm{R}(t)$), it still retains the following useful property:
\begin{framed}
\begin{lemma} \label{pi_class_lemma}
$\pi'\in \Pi^*$.
\end{lemma}
\end{framed}
\begin{proof}
See Appendix \eqref{pi_class_lemma_proof}. 
\end{proof}
The above lemma states that the policy $\pi'$ inherits the in-order delivery property and the in-neighbour packet delivery constraint of the policy-space $\Pi^*$.\\
 
From Eqn. \eqref{weight_comp} it follows that, computation of link-weights $\{W_{ij}(t), i\in \partial^{\mathrm{in}}(j)\}$ by node $j$ requires packet-state information of the nodes that are located within $2$-hops from the node $j$. Thus, it is natural to expect that with an ergodic state-update process, the weights $W_{ij}'(t)$, computed from the imperfect packet-state information, will not differ too much from the true weights $W_{ij}(t)$, on the average. Indeed, we can bound the difference between the link-weights $W'_{ij}(t)$, used by policy $\pi'$ and the true link-weights $W_{ij}(t)$, as follows
\begin{framed}
\begin{lemma} \label{update_lemma}
There exists a finite constant $C$ such that, the expected weight $W_{ij}'(t)$ of the link $(ij)$, locally computed by the node $j$ using the random update process, differs from the true link-weight $W_{ij}(t)$ by at most $C$, i.e.
\begin{eqnarray}
|\mathbb{E}W_{ij}'(t) -W_{ij}(t)| \leq C
\end{eqnarray}
The expectation above is taken with respect to the random packet-state update process.
\end{lemma}
\end{framed}
\begin{proof}
See Appendix \eqref{update_lemma_proof}
\end{proof}

From lemma \eqref{update_lemma} it follows that the policy $\pi'$, in which link-weights are computed using imperfect packet-state information is also a throughput-optimal broadcast policy for a wireless DAG. Its proof is very similar to the proof of Theorem \eqref{pi_star_optimality}. However, since the policy $\pi'$ makes scheduling decision using $\bm{W}'(t)$, instead of $\bm{W}(t)$, we need to appropriately bound the differences in drift using the Lemma \eqref{update_lemma}. The technical details are provided in Appendix \eqref{pi_dash_optimality_theorem_proof}. 
\begin{framed}
\begin{theorem} \label{pi_dash_optimality_theorem}
The policy $\pi'$ is a throughput-optimal broadcast algorithm in a time-varying wireless $\mathrm{DAG}$. 
\end{theorem}
\end{framed}

\section{Numerical Simulation} \label{simulation_section}
We numerically simulate the performance of the proposed dynamic broadcast-policy on the $3\times 3$ grid network, shown in Figure \ref{grid_network_fig}. All links are assumed to be of unit capacity. Wireless link activations are subject to primary interference constraints, i.e., at every slot, we may activate a subset of links 
\begin{figure}
\centering
\includegraphics[scale=0.75]{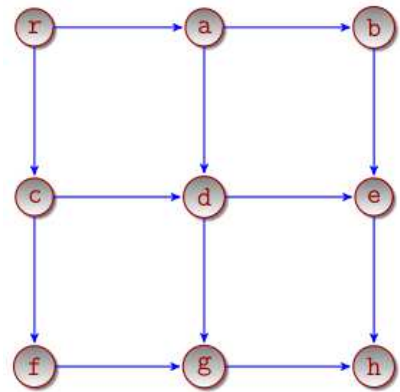}
\caption{A $3\times 3$ grid network.}
\label{grid_network_fig}
\end{figure}
which form a \emph{Matching} \cite{west2001introduction} of the underlying topology. External packets arrive at the source node $\texttt{r}$ according to a Poisson process of rate $\lambda$ packets per slot. The following proposition shows that, the broadcast capacity $\lambda^*_{\mathrm{stat}}$ of the static $3\times 3$ wireless grid (i.e., when all links are ON with probability $1$ at every slot) is $\frac{2}{5}$. 
\begin{framed}
\begin{proposition} \label{grid_network_capacity}
The broadcast-capacity $\lambda^*_{\mathrm{stat}}$ of the static $3\times 3$ wireless grid-network in Figure \ref{grid_network_fig} is $\frac{2}{5}$. 
\end{proposition}
\end{framed}
See Appendix \eqref{grid_network_capacity_proof} for the proof. \\
In our numerical simulation, the time-variation of the network is modeled as follows: link-states are assumed to evolving in an i.i.d. fashion; each link is ON with probability $p$ at every slot, independent of everything else. Here $0<p\leq 1$ is the \emph{connectivity-parameter} of the network. Thus, for $p=1$ we recover the static network model of \cite{sinha_DAG}. We also assume that the nodes have imperfect packet-state information as in Section \ref{intermittent_connectivity}. Hence, two nodes $i$ and $j$ can directly exchange packet state-information, only when the link $(i,j)$ (if any) is ON. \\
The average broadcast-delay $D^{\pi'}_p(\lambda)$ is plotted in Figure \ref{lambda_delay_plot} as a function of the packet arrival rate $\lambda$. The broadcast-delay of a packet is defined as the number of slots the packet takes to reach all nodes in the network after its arrival.  Because of the throughput-optimality of the policy $\pi'$ (Theorem \eqref{pi_dash_optimality_theorem}), the broadcast-capacity $\lambda^*(p)$ of the network, for a given value of $p$, may be empirically evaluated from the $\lambda$-intercept of vertical asymptote of the $D^{\pi'}_p(\lambda)-\lambda$ curve. 

 As evident from the plot, for $p=1$, the proposed dynamic algorithm achieves all broadcast rates below $\lambda^*_{\mathrm{stat}}=\frac{2}{5}=0.4$. This shows the throughput-optimality of the algorithm $\pi'$. \\
It is evident from the Figure \ref{lambda_delay_plot} that the broadcast capacity $\lambda^*(p)$ is non-decreasing in the connectivity-parameter $p$, i.e., $\lambda^*(p_1)\geq \lambda^*(p_2)$ for $p_1\geq p_2$. We observe that, with i.i.d. connectivity, the capacity bounds given in Lemma \eqref{bd_lemma} are not tight, in general. Hence the lower-bound of $p\lambda^*_{\mathrm{stat}}$ is a pessimistic estimate of the actual broadcast capacity $\lambda^*(p)$ of the DAG. The plot also reveals that, $D^{\pi'}_p(\lambda)$ is non-decreasing in $\lambda$ for a fixed $p$ and non-increasing in $p$ for a fixed $\lambda$, as expected.  

\begin{figure} [!ht] 
\centering
\begin{overpic}[width=0.5\textwidth]{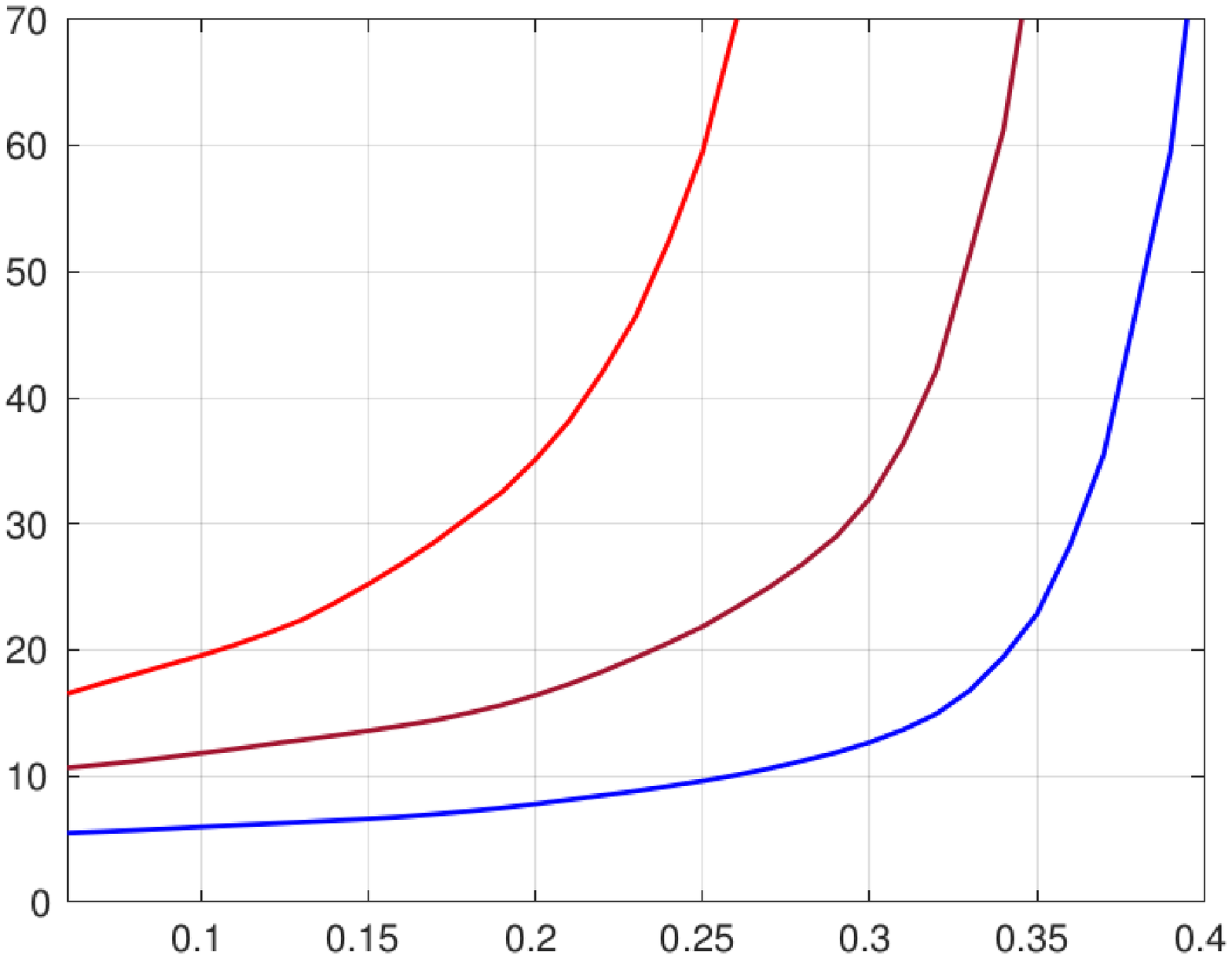}
\put(35,-1){\texttt{Packet Arrival Rate $\lambda$}}
\put(2,10){\rotatebox{90}{\texttt{Average Broadcast-Delay $D^{\pi'}_p(\lambda)$}}}
\put(73,18){$p=1$}
\put(61,27){$p=0.6$}
\put(53,47){$p=0.4$}
\end{overpic}
\caption{Plot of average broadcast-delay $D^{\pi'}_p(\lambda)$, as a function of the packet arrival rates $\lambda$. The underlying wireless network is the $3\times 3$ grid, shown in Figure \ref{grid_network_fig}, with primary interference constraints.}
\label{lambda_delay_plot}
\end{figure}

\section{Conclusion} \label{conclusion}
In this paper we studied the problem of throughput-optimal broadcasting in wireless directed acyclic networks with point-to-point links and time-varying connectivity. We characterized the broadcast-capacity of such networks and derived efficient algorithms for computing it, both exactly and approximately. Next, we proposed a throughput-optimal broadcast policy for such networks. This algorithm does not require any spanning tree to be maintained and operates based on local information, which is updated sporadically. The algorithm is robust and does not require statistics of the arrival or the connectivity process, thus making it useful for mobile wireless networks. The theoretical results are supplemented with illustrative numerical simulations. Future work would be to remove the restriction of the directed acyclic topology. It would also be interesting design broadcast algorithms for wireless networks with point-to-multi-point links.
\bibliographystyle{abbrv}
\bibliography{MIT_broadcast_bibliography}
\section{Appendix}
\subsection{Proof of Theorem \ref{algo_comp}} \label{Algo_proof} 
Under the primary interference constraint, the set of feasible activations of the graphs are \emph{matchings} \cite{west2001introduction}. To solve for the optimal broadcast capacity in a time-varying network, first we recast Eqn. \eqref{capacity_expr} as an LP. Although this LP has exponentially many constraints, using a well-known separation oracle for matchings, we show how to solve this LP in strongly-polynomial time via the ellipsoid algorithm \cite{bertsimas1997introduction}.\\

 For a subset of edges $E'\subset E$, let $\chi^{E'}$ be the incidence vector, where $\chi^{E'}(e)=1$ if $e \in E'$ and is zero otherwise. Let
\begin{eqnarray*} 
&&\mathcal{P}_\text{matching}(\mathcal{G}(V,E))= \\
&&\textbf{convexhull}(\{\chi^{M}|M \text{ is a matching in } G(V,E)\})
\end{eqnarray*}
We have the following classical result by Edmonds \cite{schrijver2003combinatorial}.
\begin{framed}
\begin{theorem}\label{match_poly}
The set $\mathcal{P}_\textrm{matching}(\mathcal{G}(V,E))$ is characterized by the set of all $\bm{\beta} \in \mathbb{R}^{|E|}$ such that : 
\begin{eqnarray}\label{mat_const}
\beta_e &\geq& 0 \hspace{10pt} \forall e \in E  \\
\sum_{e \in \partial^{\mathrm{in}}(v)\cup\partial^{\mathrm{out}}(v)} \beta_e &\leq& 1 \hspace{10pt} \forall v \in V \nonumber\\
\sum_{e\in E[U]}\beta_e &\leq& \frac{|U|-1}{2}; \hspace{10pt} U\subset V, \hspace{5pt}|U| \text{ odd} \nonumber 
\end{eqnarray}
\end{theorem}
\end{framed}
Here $E[U]$ is the set of edge with both end points in U. 

Thus, following Eqn. \eqref{capacity_expr}, the broadcast capacity of a DAG can be obtained by the following LP :
\begin{eqnarray} \label{LP_match}
\max \lambda 
\end{eqnarray}
Subject to, 
\begin{eqnarray}
\lambda &\leq& \sum_{e\in \partial^\mathrm{in}(v)} c_{e}\big(\sum_{\sigma \in \Xi}p(\sigma) \beta_{\sigma,e}\big), \hspace{2pt} \forall v \in V\setminus \{r\}\label{b_c_constraint}\\
\bm{\beta}_\sigma &\in& \mathcal{P}_\text{matching}(\mathcal{G}(V,E_\sigma)), \hspace{10pt} \forall \sigma \in \Xi \label{matching_constr}
\end{eqnarray}
The constraint corresponding to $\sigma \in \Xi$ in \eqref{matching_constr} refers to the set of linear constraints given in Eqn.\eqref{mat_const} corresponding to the graph $\mathcal{G}(V,E_\sigma)$, for each $\sigma \in \Xi$.\\
Invoking the equivalence of optimization and separation due to the ellipsoid algorithm \cite{bertsimas1997introduction}, it follows that the LP \eqref{LP_match} is solvable in poly-time, if there exists an efficient separator-oracle for the set of constraints \eqref{b_c_constraint} and \eqref{matching_constr}. With our assumption of polynomially many network configurations $|\Xi|(n)$,  there are only linearly many constraints ($n-1$, to be precise) in \eqref{b_c_constraint} with polynomially many variables in each constraint. Thus the  set of constraints \eqref{b_c_constraint} can be separated efficiently. 
Next we invoke a classic result from the combinatorial-optimization literature which shows the existence of efficient separators for the matching polytopes. 
\begin{framed}
\begin{theorem}{\cite{schrijver2003combinatorial}}
There exists a strongly poly-time algorithm, that given $\mathcal{G}=(V,E)$  and $\bm{\beta} : E \to \mathbb{R}^{|E|}$ determines if $\bm{\beta}$ satisfies \eqref{mat_const} or outputs an inequality from \eqref{match_poly} that is violated by $\bm{\beta}$.
\end{theorem}
\end{framed}
Hence, there exists an efficient separator for each of the constraints in \eqref{match_poly}. Since there are only polynomially many network configurations, this directly leads to Theorem \ref{algo_comp}. 
\newpage 
\subsection{Proof of Lemma \ref{conv_hull}} 
\begin{proof}
Fix a time $T$. For each configuration $\sigma \in \Xi$, let $\{t_{\sigma,i}\}_{i=1}^{T_\sigma}$ be the index of the time-slots up to time $T$ such that $\bm{\sigma}(t)=\bm{\sigma}$. Clearly we have,
\begin{eqnarray}
\sum_{\sigma \in \Xi} T_\sigma = T
\end{eqnarray}
Hence, we can rewrite 
\begin{eqnarray} \label{sum1}
\frac{1}{T}\sum_{t=1}^{T} \bm{s}^\pi (t, \bm{\sigma}(t)) = \sum_{\sigma \in \Xi} \frac{T_\sigma}{T} \frac{1}{T_\sigma}\sum_{i=1}^{T_\sigma} \bm{s}^\pi(t_{\sigma,i},\bm{\sigma})
\end{eqnarray}
Hence,
\begin{eqnarray} \label{main_eqn1}
\bm{u}\cdot\bigg(\frac{1}{T} \sum_{t=1}^{T}\bm{s}^{\pi}(t,\bm{\sigma}(t))\bigg) = \sum_{\sigma \in \Xi} \frac{T_\sigma}{T} \bm{u}\cdot \bigg(\frac{1}{T_\sigma}\sum_{i=1}^{T_\sigma} \bm{s}^\pi(t_{\sigma,i},\bm{\sigma})\bigg)
\end{eqnarray}
Since the process $\bm{\sigma}(t)$ is stationary ergodic, we have 
\begin{eqnarray}
\lim_{T \to \infty} \frac{T_\sigma}{T}=p(\sigma), \hspace*{5pt}\text{w.p. 1} \hspace{5pt} \forall \sigma \in \Xi
\end{eqnarray}
Using countability of $\Xi$ and invoking the union bound, we can strengthen the above conclusion as follows 
\begin{eqnarray}
\lim_{T \to \infty} \frac{T_\sigma}{T}=p(\sigma), \hspace*{5pt}\forall \sigma \in \Xi,\hspace{5pt} \text{w.p.}\hspace*{5pt} 1 
\end{eqnarray}
Hence from Eqn. \eqref{main_eqn1} we have,
\begin{eqnarray*} \label{main_eq22}
&&\min_{\bm{u} \in U}\liminf_{T \nearrow \infty} \bm{u}\cdot\bigg(\frac{1}{T} \sum_{t=1}^{T}\bm{s}^{\pi}(t,\bm{\sigma}(t))\bigg) \\
&=& \min_{\bm{u} \in U} \sum_{\sigma \in \Xi} p(\bm{\sigma}) \liminf_{T\to \infty} \bm{u}\cdot \bigg(\frac{1}{T_\sigma}\sum_{i=1}^{T_\sigma} \bm{s}^\pi(t_{\sigma,i},\bm{\sigma})\bigg), \hspace*{10pt}\text{w.p.} \hspace*{5pt}1
\end{eqnarray*}
Since $p(\sigma)>0, \forall \sigma \in \Xi$, the above implies that $T_\sigma \nearrow \infty \text{ as } T\nearrow \infty \forall \sigma, \hspace*{5pt}\text{w.p.}1$. 
In the rest of the proof we will concentrate on a typical sample path $\{\bm{\sigma}(t)\}_{t\geq 1}$ having the above property.\\
For each $\sigma \in \Xi$, define the sequence $\{\bm{\zeta}_{\sigma, T_\sigma}^{\pi}\}_{T_\sigma \geq 1}$
\begin{eqnarray}
\bm{\zeta}_{\sigma, T_\sigma}^{\pi} = \frac{1}{T_\sigma}\sum_{i=1}^{T_\sigma}\bm{s}^{\pi}(t_{\sigma,i},\bm{\sigma})
\end{eqnarray}
Since $\bm{s}^{\pi}(t_{\sigma,i}, \bm{\sigma}) \in \mathcal{M}_\sigma$ for all $i\geq 1$, convexity of the set $\mathcal{M}_\sigma$ implies that $\bm{\zeta}_{\sigma, T_\sigma}^{\pi} \in \mathcal{M}_\sigma$ for all $T_\sigma\geq 1$. Since the set $\mathcal{M}_\sigma$ is closed and bounded (and hence, compact) any sequence  in $\mathcal{M}_\sigma$ has a converging sub-sequence. Consider any set of converging sub-sequences $\{\bm{\zeta}_{\sigma, {T_{\sigma_k}}}^{\pi}\}_{k\geq 1}, \sigma \in \Xi$ such that, it achieves the following 

\begin{eqnarray*} \label{lim}
&&\min_{u \in U}  \sum_{\sigma \in \Xi}p(\sigma)\lim_{k \to \infty}\bm{u}\cdot \bm{\zeta}_{\sigma, {T_{\sigma_k}}}^{\pi} \\
&=& \min_{u \in U} \sum_{\sigma \in \Xi}p(\sigma)\liminf_{T_\sigma\to \infty} \bm{u} \cdot \bm{\zeta}_{\sigma, T_\sigma}^{\pi}.
\end{eqnarray*} 
Let us denote 
\begin{eqnarray} \label{lim2}
\lim_{k \to \infty}\bm{\zeta}_{\sigma, {T_{\sigma_k}}}^{\pi} = \bm{\beta}_\sigma^{\pi}, \hspace*{10pt}\forall \sigma \in \Xi
\end{eqnarray}
Where $\bm{\beta}_\sigma^{\pi} \in \mathcal{M}_\sigma$, since $\mathcal{M}_\sigma$ is closed. Hence combining Eqn. \eqref{main_eq22}, \eqref{lim} and Eqn. \eqref{lim2}, we have 
\begin{eqnarray*}
&&\min_{\bm{u} \in U}\liminf_{T \nearrow \infty} \bm{u}\cdot\bigg(\frac{1}{T} \sum_{t=1}^{T}\bm{s}^{\pi}(t,\bm{\sigma}(t))\bigg) \\
&=& \min_{\bm{u} \in U} \sum_{\sigma \in \Xi} p(\bm{\sigma}) \bm{u} \cdot \bm{\beta}^{\pi}_{\sigma} \\
&=& \min_{\bm{u} \in U} \bm{u}\cdot\bigg( \sum_{\sigma \in \Xi} p(\bm{\sigma})\bm{\beta}^{\pi}_{\sigma} \bigg) \hspace*{5pt}\mathrm{w.p.} 1
\end{eqnarray*}
\end{proof}

\subsection{Proof of Lemma \ref{bd_lemma}}\label{bd_proof}
\begin{proof}
\subsubsection{Proof of the Upper-bound}
Note that, for all $\sigma \in \Xi$, we have $E_\sigma \subset E$. Hence, it follows that 
\begin{eqnarray*}
	\mathcal{M}_\sigma \subset \mathcal{M}, \hspace{10pt} \forall \sigma \in \Xi 
\end{eqnarray*}
This in turn implies that 
\begin{eqnarray} \label{conv_hull_eqn}
	\beta_\sigma \in \mathrm{conv}(\mathcal{M}_\sigma) \implies \beta_\sigma \in \mathrm{conv}(\mathcal{M}) 
\end{eqnarray}
Let an optimal solution to Eqn. \eqref{capacity_expr} be obtained at $\big(\bm{\beta}^*_\sigma, \sigma \in \Xi\big) $. Then from Eqn. \eqref{conv_hull_eqn}, it follows that 
\begin{eqnarray*}
	\sum_{\sigma \in \Xi} p(\bm{\sigma})\bm{\beta}^*_\sigma \in \mathrm{conv} (\mathcal{M}) 
\end{eqnarray*}
Hence we have, 
\begin{eqnarray*}
	\max_ {\bm{\beta}_\sigma\in \mathrm{conv}(\mathcal{M}_\sigma)} \min_{j\in V\setminus \{\texttt{r}\}} \bm{u}_j \cdot \big(\sum_{\sigma \in \Xi} p(\bm{\sigma})\bm{\beta}^*_\sigma\big)\\ \leq \max_{\bm{\beta} \in \mathrm{conv}(\mathcal{M})} \min_{j\in V \setminus \{\texttt{r}\}} \bm{u}_j \cdot \bm{\beta}  
\end{eqnarray*}
Using Eqn. \eqref{whole_cap}, this shows that 
\begin{eqnarray*}
\lambda^* \leq \lambda^*_{\mathrm{stat}}
\end{eqnarray*}
This proves the upper-bound. 
\subsubsection{Proof of the Lower-bound}
Since $\mathcal{M}_\sigma \subset \mathcal{M}$,  
the expression for the broadcast-capacity \eqref{capacity_expr} may be re-written as follows:
\begin{eqnarray*}
	\lambda^*=\max_{\bm{\beta}_\sigma \in \mathcal{M}} \min_{j \in V\setminus \{\texttt{r}\}}\sum_{e \in \partial^{\mathrm{in}}(j)} c_e\big(\sum_{\sigma \in \Xi} p(\bm{\sigma}) \bm{\beta}_{\bm{\sigma}}(e) \mathbbm{1}(e \in \bm{\sigma})\big)
\end{eqnarray*}
Let $\bm{\beta}^* \in \mathcal{M}$ be the optimal activation, achieving the RHS of \eqref{whole_cap}. 
Hence we can lower-bound $\lambda^*$ as follows 
\begin{eqnarray*}
	\lambda^* &\geq& \min_{j \in V\setminus \{\texttt{r}\}}\sum_{e \in \partial^{\mathrm{in}}(j)} c_e\bm{\beta}^*(e)\big(\sum_{\sigma \in \Xi} p(\bm{\sigma}) \mathbbm{1}(e \in \bm{\sigma})\big)\\
	&\stackrel{(a)}{=}& p\min_{j \in V\setminus \{\texttt{r}\}}\sum_{e \in \partial^{\mathrm{in}}(j)} c_e\bm{\beta}^*(e)\\
	&=&p \min_{j \in V\setminus \{\texttt{r}\}} \bm{u}_j \cdot \bm{\beta}^*\\
	&\stackrel{(b)}{=}&p \lambda^*_{\mathrm{stat}}
	\end{eqnarray*}
  Equality (a) follows from the assumption \eqref{p_act} and  equality (b) follows from the characterization \eqref{whole_cap}. This proves the lower-bound.

\end{proof}

\subsection{Proof of Corollary \ref{approx-comp}} \label{approx-comp-proof}
Consider the optimal randomized-activation vector $\bm{\beta}^*\in \mathcal{M}$, corresponding to the stationary graph $\mathcal{G}(V,E)$ \eqref{whole_cap}. By Theorem \eqref{algo_comp}, $\bm{\beta}^*$ can be computed in poly-time under the primary interference constraint. Note that, by Caratheodory's theorem \cite{matouvsek2002lectures}, the support of $\bm{\beta}^*$ may be bounded by $|E|$. Hence it follows that $\lambda^*_{\mathrm{stat}}$ \eqref{whole_cap} may also be computed in poly-time.\\  
From the proof of Lemma \eqref{bd_lemma}, it follows that by randomly activating $\bm{\beta}^*$ (i.e., $\bm{\beta}_{\bm{\sigma}}(e)= \bm{\beta}^*(e) \mathbbm{1}(e \in \sigma), \forall \bm{\sigma} \in \Xi$) we obtain a broadcast-rate equal to $p\lambda^*_{\mathrm{stat}}$ where $\lambda^*_{\mathrm{stat}}$ is shown to be an upper-bound to the broadcast capacity $\lambda^*$ in Lemma \eqref{bd_lemma}. Hence it follows that $p\lambda^*_{\mathrm{stat}}$ constitutes a $p$-approximation to the broadcast capacity $\lambda^*$, which can be computed in poly-time.   

\subsection{Proof of Lemma \eqref{stability_lemma}} \label{stability_lemma_proof}

Assume that under the policy $\pi \in \Pi^*$, the virtual queues $X_{j}(t)$ are rate stable i.e., $\lim_{T\to \infty} X_j(T)/T = 0, a.s.$ for all $j$. Applying union-bound, it follows that,
\begin{equation} \label{rate_stability_of_X}
\lim_{T\to \infty} \frac{\sum_{j\neq r}X_j(T)}{T} = 0, \hspace{15pt} \text{w.p.}  \hspace{3pt}1
\end{equation}
Now consider any node $j\neq \texttt{r}$ in the network. We can construct a simple path $p (\texttt{r} = u_k \to u_{k-1} \ldots \to u_1 = j)$ from the source node $\texttt{r}$ to the node $j$ by running the following \texttt{Path construction algorithm} on the underlying graph $\mathcal{G}(V,E)$. 
 \begin{algorithm} 
\caption{$\texttt{r}\to j$ \texttt{Path Construction Algorithm}}
\begin{algorithmic}[1] 
 \REQUIRE DAG $\mathcal{G}(V,E)$, node $j\in V$
 \STATE $i \gets 1$
 \STATE $u_i\gets j$
 \WHILE{$u_i \neq r$} 
 \STATE $u_{i+1} \gets i_t^*(u_i)$;
 \STATE $i \gets i+1$
 \ENDWHILE
 \end{algorithmic}
 \end{algorithm} 
 
At time $t$, the algorithm chooses the parent of a node $u_i$ in the path $p$ as the one that has the least relative packet deficit as compared to $u_i$ (i.e. $u_{i+1}=i_t^*(u_i)$). Since the underlying graph $\mathcal{G}(V,E)$ is a connected DAG (i.e., there is a path from the source to every other node in the network), the above path construction algorithm always terminates with a path $p(\texttt{r}\to j)$. Note that the output path of the algorithm varies with time.\\
 The number of distinct packets received by node $j$ up to time $T$ can be written as a telescoping sum of relative packet deficits along the path $p$, i.e.,
\begin{align}
R_j(T) &= R_{u_1}(T) \notag \\
&= \sum_{i=1}^{k-1}\big(R_{u_i}(T)-R_{u_{i+1}}(T)\big) +R_{u_k}(T)  \notag \\
&= -\sum_{i=1}^{k-1} X_{u_i}(T) + R_{\texttt{r}}(T) \notag \\
&\stackrel{(a)}{=} -\sum_{i=1}^{k-1} X_{u_i}(T) + \sum_{t=0}^{T-1} A(t), \label{eq:116}
\end{align}
where the equality (a) follows the observation that 
\[
X_{u_{i}}(T) = Q_{u_{i+1}u_{i}}(T) = R_{u_{i+1}}(T) - R_{u_{i}}(T).
\]
Since the variables $X_i(t)$'s are non-negative, we have $\sum_{i=1}^{k-1} X_{u_{i}}(t) \leq \sum_{j\neq r} X_{j}(t)$. Thus, for each node $j$
\[
\frac{1}{T}\sum_{t=0}^{T-1} A(t) -\frac{1}{T}\sum_{j\neq r} X_{j}(T) \leq \frac{1}{T} R_j(T) \leq \frac{1}{T}\sum_{t=0}^{T-1} A(t).
\]
Taking limit as $T\to \infty$ and using the strong law of large numbers for the arrival process and Eqn. \eqref{rate_stability_of_X}, we have
\[
 \lim_{T\to \infty} \frac{R_j(T)}{T} = \lambda, \, \forall j. \hspace{10pt} \text{w.p.}\hspace*{5pt}1
\]
This concludes the proof.


\subsection{Proof of Lemma \eqref{stabilizing_X_lemma}} \label{stabilizing_X_lemma_proof}

%
%
We begin with a preliminary lemma. 
\begin{lemma} \label{algebra}
If we have
\begin{equation} \label{eq:108}
Q(t+1)\leq  (Q(t)-\mu(t))^+ + A(t) 
\end{equation}
where all the variables are non-negative and $(x)^+ = \max\{x,0\}$, then
\[
Q^2(t+1) - Q^2(t) \leq \mu^2(t) + A^2(t) + 2Q(t)(A(t)-\mu(t)).
\]
\end{lemma}
\begin{proof}
Squaring both sides of Eqn. \eqref{eq:108} yields,
\begin{align*}
&Q^2(t+1) \\
&\leq  \big((Q(t)-\mu(t))^+\big)^2 + A^2(t) + 2 A(t)(Q(t)-\mu(t))^+\\
&\leq  (Q(t)-\mu(t))^2 + A^2(t) + 2 A(t)Q(t),
\end{align*}
where we use the fact that $x^2 \geq {(x^+)}^2$, $Q(t) \geq 0$, and $\mu(t) \geq 0$. Rearranging the above inequality finishes the proof.
\end{proof}
Applying Lemma~\ref{algebra} to the dynamics~\eqref{dyn} of $X_{j}(t)$ yields, for each node $j\neq r$,
\begin{eqnarray} \label{eq:109}
X_j^2(t+1) - X_j^2(t) 
\leq  B(t) + \\
 2 X_j(t) \big(\sum_{m\in V}\mu_{mi_{t}^*}(t)-\sum_{k\in V} \mu_{kj}(t)\big),
\end{eqnarray}
where $B(t)\leq c^2_{\max}+ \max\{a^2(t),c^2_{\max}\} \leq  (a^2(t) + 2c^2_{\max})$, $a(t)$ is the number of exogenous packet arrivals in a slot, and $c_{\max} \triangleq \max_{e\in E} c_e$ is the maximum capacity of the links.
We assume the arrival process $a(t)$ has bounded second moments; thus, there exists a finite constant $B>0$ such that $\mathbb{E}[B(t)] \leq \mathbb{E}\big(a^2(t)\big) + 2c^2_{\max} < B$.

We define the quadratic Lyapunov function $L(\bm{X}(t)) = \sum_{j\neq r} X_j^2(t)$. From~\eqref{eq:109}, the one-slot Lyapunov drift $\Delta(\bm{X}(t))$, \emph{conditioned} on the current network-configuration $\bm{\sigma}(t)$ yields
\begin{eqnarray} \label{drift2}
&&\Delta(\bm{X}(t)|\bm{\sigma}(t)) \triangleq  \mathbb{E}[L(\bm{X}(t+1) - L(\bm{X}(t)) \mid \bm{X}(t), \bm{\sigma}(t)] \nonumber \\
&=& \mathbb{E}\big[\sum_{j\neq r} \big(X_j^2(t+1) - X_j^2(t) \big) \mid \bm{X}(t),\bm{\sigma}(t)\big] \notag \\
&\leq& B|V| +2  \sum_{j\neq r} X_{j}(t) \mathbb{E}\big[\sum_{m\in V}\mu_{mi_{t}^*}(t)\\
&&-\sum_{k\in V} \mu_{kj}(t) \mid \bm{X}(t), \bm{\sigma}(t)\big] \notag \\
&=& B|V| - 2 \sum_{(i,j)\in E} \mathbb{E}[\mu_{ij}(t)\mid\bm{X}(t), \bm{\sigma}(t)] \big( X_j(t)\\
&& - \sum_{k\in K_{j}(t)} X_k(t) \big) \notag \\
&=& B|V|- 2 \sum_{(i,j)\in E}  W_{ij}(t) \mathbb{E}[\mu_{ij}(t)\mid \bm{X}(t),\bm{\sigma}(t)] \label{drift3}
\end{eqnarray}
The broadcast-policy $\pi^{*}$ is chosen to minimize the upper-bound of \emph{conditional-drift}, given on the right-hand side of \eqref{drift3} among all policies in $\Pi^{*}$.

Next, we construct a randomized scheduling policy $\pi^{\text{RAND}} \in\Pi^{*}$. Let $\bm{\beta}_{\bm{\sigma}}^*\in\text{conv}({\mathcal{M}_\sigma})$ be the part of an optimal solution corresponding to $\bm{\sigma}(t)\equiv \bm{\sigma}$  given by Eqn.~\ref{bc_ob}.
From Caratheodory's theorem~\cite{matouvsek2002lectures}, there exist at most $(|E|+1)$ link-activation vectors $\bm{s}_k\in \mathcal{M}_\sigma$ and the associated non-negative scalars $\{\alpha_k^\sigma\}$ with $\sum_{k=1}^{|E|+1}\alpha_k^\sigma=1$, such that 
\begin{equation} \label{beta_star}
\bm{\beta}^*_\sigma= \sum_{k=1}^{|E|+1} \alpha_k^\sigma \bm{s}_k^\sigma.
\end{equation}
Define the average (unconditional) activation vector 
\begin{eqnarray}
\bm{\beta}^*=\sum_{\sigma \in \Xi} p(\sigma)\bm{\beta}^*_\sigma
\end{eqnarray}
Hence, from Eqn. \eqref{bc_ob} we have,
\begin{equation} \label{bc_bound}
\lambda^* \leq \min_{\text{$U$: a proper cut}} \sum_{e\in E_{U}} c_{e} \beta_{e}^{*}.
\end{equation}
Suppose that the exogenous packet arrival rate $\lambda$ is strictly less than the broadcast capacity $\lambda^*$. There exists an $\epsilon >0$ such that $\lambda +\epsilon \leq \lambda^{*}$. From~\eqref{bc_bound}, we have
\begin{equation} \label{eq:113}
\lambda+\epsilon \leq \min_{\text{$U$: a proper cut}} \sum_{e\in E_{U}} c_{e} \beta_{e}^{*}.
\end{equation}
For any network node $v\neq r$, consider the proper cuts $U_{v} = V\setminus \{v\}$. Specializing the bound in \eqref{eq:113} to these cuts, we have
\begin{equation} \label{capacity_exceeding}
\lambda + \epsilon \leq  \sum_{e\in E_{U_{v}}} c_{e} \beta_{e}^{*}, \ \forall v\neq r.
\end{equation}
Since the underlying network topology $\mathcal{G}=(V, E)$ is a DAG, there exists a topological ordering  of the network nodes so that: $(i)$ the nodes can be labelled serially as $\{v_{1}, \ldots, v_{|V|}\}$, where $v_{1}=r$ is the source node with no in-neighbours and $v_{|V|}$ has no outgoing neighbours and $(ii)$ all edges in $E$ are directed from $v_i \to v_j$, $i<j$ ~\cite{algorithms};  From~\eqref{capacity_exceeding}, we define $q_{l}\in[0, 1]$ for each node $v_{l}$ such that 
\begin{equation} \label{q_prob}
q_{l}\, \sum_{e\in E_{U_{v_{l}}}} c_{e} \beta_{e}^{*} = \lambda + \epsilon \frac{l}{|V|},\  l=2, \ldots , |V|.
\end{equation}
Consider the randomized broadcast policy $\pi^{\text{RAND}} \in \Pi^{*}$ working as follows: 
\begin{framed}
\textbf{Stationary Randomized Policy $\pi^{\text{RAND}}$:}\\
(i) If the observed network-configuration at slot $t$ is $\bm{\sigma}(t)=\sigma$, the policy $\pi^{\text{RAND}}$ \textbf{selects} \footnote{\textbf{Selected} does not necessarily mean \textbf{activated}, see point (ii)} the feasible activation set $\bm{s}_k^\sigma$ with probability $\alpha_k^\sigma$;  \\
(ii) For each incoming selected link $e = (\cdot, v_{l})$ to node $v_{l}$ such that $s_{e}(t)=1$, the link $e$ is \textbf{activated} independently with probability $q_{l}$;\\
 (iii) \textbf{Activated} links (note, not necessarily all the \emph{selected} links) are used to forward packets, subject to the constraints that define the policy class $\Pi^{*}$ (i.e., in-order packet delivery and that a network node is only allowed to receive packets that have been received by all of its in-neighbors).
\end{framed}
  Note that this stationary randomized policy $\pi^{\text{RAND}}$ operates independently of the state of received packets in the network, i.e.,  $\bm{X}(t)$. However it depends on the current network-configuration $\bm{\sigma}(t)$.  Since each network node $j$ is relabelled as $v_{l}$ for some $l$, from~\eqref{q_prob} we have, for each node $j\neq r$, the total expected incoming transmission rate to the node $j$ under the policy $\pi^{\text{RAND}}$, averaged over all network states $\bm{\sigma}$ satisfies
\begin{align} 
\sum_{i: (i, j)\in E}\mathbb{E}[\mu^{\pi^{\text{RAND}}}_{ij}(t)\mid\bm{X}(t)] &=\sum_{i: (i,j)\in E} \mathbb{E}[\mu^{\pi^{\text{RAND}}}_{ij}(t)]  \notag \\
&= q_{l}\, \sum_{e\in E_{U_{v_{l}}}} c_{e} \beta_{e}^{*} \notag \\
&=\lambda + \epsilon \frac{l}{|V|}. \label{rate_comp1}
\end{align}
Equation~\eqref{rate_comp1} shows that the randomized policy $\pi^{\text{RAND}}$ provides each network node $j\neq r$ with the total expected incoming rate strictly larger than the packet arrival rate $\lambda$ via proper random link activations conditioned on the current network configuration. According to our notational convention,  we have
\begin{equation} \label{rate_comp2}
\sum_{i:(i,r)\in E} \mathbb{E}[\mu^{\pi^{\text{RAND}}}_{ir}(t)\mid\bm{X}(t)] = \mathbb{E}[\sum_{i:(i,r)\in E} \mu^{\pi^{\text{RAND}}}_{ir}(t)] = \lambda.
\end{equation}
From~\eqref{rate_comp1} and~\eqref{rate_comp2}, if node $i$ appears before node $j$ in the aforementioned topological ordering, i.e., $i = v_{l_{i}} < v_{l_{j}} = j$ for some $l_{i} < l_{j}$, then
\begin{align} 
&\sum_{k:(k,i)\in E}\mathbb{E}[\mu^{\pi^{\text{RAND}}}_{ki}(t)]- \sum_{k:(k,j)\in E}\mathbb{E}[\mu^{\pi^{\text{RAND}}}_{kj}(t)]  \notag \\
&\leq -\frac{\epsilon}{|V|}. \label{rate_comparison_final}
\end{align}
The above inequality will be used to show the throughput optimality of the policy $\pi^*$.\\
The drift inequality~\eqref{drift2} holds for any policy $\pi \in \Pi^*$. The broadcast policy $\pi^{*}$ observes the states $(\bm{X}(t), \bm{\sigma}(t))$ and and seek to \emph{greedily} minimize the upper-bound of drift \eqref{drift3} at every slot. Comparing the actions taken by the policy $\pi^{*}$ with those by the randomized policy $\pi^{\text{RAND}}$ in slot $t$ in~\eqref{drift2}, we have
\begin{align}
&\Delta^{\pi^*}(\bm{X}(t)|\bm{\sigma}(t))\\
& \leq B|V|- 2 \sum_{(i,j)\in E}\mathbb{E}\big[\mu^{\pi^{*}}_{ij}(t) \mid\bm{X}(t), \bm{\sigma}(t)] W_{ij}(t) \notag \\
&\leq B|V|- 2 \sum_{(i,j)\in E}\mathbb{E}\big[\mu^{\pi^{\text{RAND}}}_{ij}(t) \mid\bm{X}(t), \bm{\sigma}(t)] W_{ij}(t) \notag \\
&\stackrel{(*)}{=}  B|V|- 2 \sum_{(i,j)\in E}\mathbb{E}\big[\mu^{\pi^{\text{RAND}}}_{ij}(t) \mid \bm{\sigma}(t)] W_{ij}(t) \notag \\
\end{align}
Taking Expectation of both sides w.r.t. the stationary-process $\bm{\sigma}(t)$ and rearranging, we have
\begin{align}
&\Delta^{\pi^*}(\bm{X}(t))\\
& \leq B|V|- 2 \sum_{(i,j)\in E}\mathbb{E}\big[\mu^{\pi^{\text{RAND}}}_{ij}(t) ] W_{ij}(t) \notag \\
&\leq  B|V| +2 \sum_{j\neq r} X_j(t) \bigg(\sum_{m\in V}\mathbb{E}\big[\mu^{\pi^{\text{RAND}}}_{mi_{t}^*}(t)\big]  -\sum_{k\in V} \mathbb{E}\big[\mu^{\pi^{\text{RAND}}}_{kj}(t)\big]\bigg) \notag \\
&\leq B|V| - \frac{2\epsilon}{|V|} \sum_{j\neq r} X_{j}(t). \label{eq:114}
\end{align}
Note that $i_{t}^*= \arg \min_{i\in \text{In}(j)} Q_{ij}(t)$ for a given node $j$. Since node $i_{t}^{*}$ is an in-neighbour of node $j$, $i_{t}^{*}$ must lie before $j$ in any topological ordering of the DAG. Hence, the last inequality of~\eqref{eq:114} follows directly from ~\eqref{rate_comparison_final}. Taking expectation in~\eqref{eq:114} with respect to $\bm{X}(t)$, we have
\[
\mathbb{E}\big[L(\bm{X}(t+1))\big]-\mathbb{E}\big[L(\bm{X}(t))\big] \leq B|V| -\frac{2\epsilon}{|V|}\mathbb{E}||\bm{X}(t)||_1,
\]
where $||\cdot ||_1$ is the $\ell_1$-norm of a vector. Summing the above inequality over $t=0, 1,2,\ldots T-1$ yields
\[
\mathbb{E}\big[L(\bm{X}(T))\big]-\mathbb{E}\big[L(\bm{X}(0))\big] \leq B|V|T -\frac{2\epsilon}{|V|}\sum_{t=0}^{T-1}\mathbb{E}||\bm{X}(t)||_1.
\]
Dividing the above by $2T\epsilon/|V|$ and using $L(\bm{X}(t))\geq 0$, we have 
\begin{eqnarray*}
\frac{1}{T}\sum_{t=0}^{T-1}\mathbb{E}||\bm{X}(t)||_1 \leq \frac{B|V|^2}{2\epsilon} + \frac{|V|\,\mathbb{E}[L(\bm{X}(0))]}{2T\epsilon}
\end{eqnarray*}
Taking a $\limsup$ of both sides yields
\begin{eqnarray} \label{strong_stability}
\limsup_{T \to \infty}\frac{1}{T}\sum_{t=0}^{T-1} \sum_{j\neq r} \mathbb{E}[X_{j}(t)]  \leq \frac{B|V|^2}{2\epsilon}
\end{eqnarray} 
which implies that all virtual-queues $X_{j}(t)$ are strongly stable \cite{neely2010stochastic}. Strong stability of $X_{j}(t)$ implies that all virtual queues $X_{j}(t)$ are rate stable~\cite[Theorem~$2.8$]{neely2010stochastic}.

\subsection{Proof of Lemma \eqref{pi_class_lemma}} \label{pi_class_lemma_proof}
\begin{proof} 
Recall the definition of the policy-space $\Pi^*$. For every node $i$, since $R_i(t)$ is a \emph{non-decreasing function} of $t$, if a packet $p$ is allowed to be transmitted to a node $j$ at time slot $t$, by the policy $\pi'$, it is certainly allowed to be transmitted by the policy $\pi$. This is because  $R_i'(t) \leq R_i(t), \forall j \in \partial^{\mathrm{out}}(i)$ and hence outdated state-information may only prevent transmission of a packet $p$ at a time $t$, which would otherwise be allowed by the policy $\pi^*$. As a result, the policy $\pi'$ can never transmit a packet to node $j$ which is not present at all in-neighbours of the node $j$. This shows that $\pi' \in \Pi^*$. 
\end{proof}

\subsection{Proof of Lemma \eqref{update_lemma} } \label{update_lemma_proof}
Consider the packet-state update process at node $j$. Since the capacity of the links are bounded by $c_{\max}$, from Eqn. \eqref{delayed_update} and the fact that $R_i(t)$ is non-decreasing, we have 
\begin{eqnarray}
R_i(t) - Tc_{\max}\leq  R_i'(t) \leq R_i(t), \hspace{5pt} \forall i \in \partial^{\mathrm{in}}(j)
\end{eqnarray}
Hence, from Eq. \eqref{X-def}, it follows that
\begin{eqnarray}
X_j(t) - Tc_{\max} \leq X_j'(t) \leq X_j(t) 
\end{eqnarray}
i.e., 
\begin{eqnarray} \label{eq1234}
X_j(t) - c_{\max}\mathbb{E}T \leq \mathbb{E}X_j'(t) \leq X_j(t)
\end{eqnarray}
Where the expectation is with respect to the random update process at the node $j$. 
In a similar fashion,  since every in-neighbour $i$ of a node $k \in \partial^{\mathrm{out}}(j)$, is at most $2$-hop away from the node $j$, we have 
\begin{eqnarray*}
R_i(t)-Tc_{\max} \leq R_i'(t) \leq R_i(t) 
\end{eqnarray*}
Also,
\begin{eqnarray*}
R_k(t) - Tc_{\max} \leq R_k'(t) \leq R_k(t) 
\end{eqnarray*}
It follows that for all $i \in \partial^{\mathrm{in}}(k)$
\begin{eqnarray*}
(R_i(t)-R_k(t)) -Tc_{\max} &\leq & R_i'(t)-R_k'(t) \\
 &\leq & (R_i(t)-R_k(t)) + Tc_{\max} 
\end{eqnarray*}
Hence,
\begin{eqnarray*}
X_k(t)-Tc_{\max} \leq X_k'(t) \leq X_k(t)+Tc_{\max} 
\end{eqnarray*}
Again taking expectation w.r.t. the random packet-state update process, 
\begin{eqnarray} \label{eq1233}
X_k(t)-c_{\max} \mathbb{E}T \leq \mathbb{E} X_k'(t) \leq X_k(t) +c_{\max} \mathbb{E}T
\end{eqnarray}
Combining Eqns \eqref{eq1234} and \eqref{eq1233} using Linearity of expectation and using Eqn. \eqref{weight_comp} we have
\begin{eqnarray*}
 -nc_{\max} \mathbb{E}T+W_{ij}(t)   \leq \mathbb{E}W'_{ij}(t) \leq W_{ij}(t) +nc_{\max} \mathbb{E}T
\end{eqnarray*}
Thus the lemma \eqref{update_lemma} follows with $C\equiv nc_{\max} \mathbb{E}T <\infty $.

\subsection{Proof of Theorem \eqref{pi_dash_optimality_theorem}} \label{pi_dash_optimality_theorem_proof}
To prove throughput-optimality of Theorem \eqref{pi_dash_optimality_theorem}, we work with the same Lyapunov function $L(\bm{X}(t))=\sum_{j \neq \texttt{r}}X_j^2(t) $ as in Theorem \eqref{pi_star_optimality} and follow the same steps until Eqn. \eqref{drift3} to obtain the following upper-bound on conditional drift
\begin{eqnarray}
&&\Delta^{\pi'}(\bm{X}(t)|\bm{X}(t),\bm{X}'(t),\bm{\sigma}(t)) \nonumber \\
&\leq&  B|V|- 2\sum_{(i,j)\in E} W_{ij}(t) \mathbb{E}(\mu^{\pi'}_{ij}(t)| \bm{X}(t), \bm{X}'(t), \bm{\sigma}(t)) \nonumber \\
 \label{eq:rand_drift} 
\end{eqnarray}
Since the policy $\pi'$ makes scheduling decision based on the \emph{locally computed} weights $W'_{ij}(t)$, by the definition of the policy $\pi'$,  we have for any policy $\pi \in \Pi$:
\begin{eqnarray} \label{w2}
\sum_{(i,j)\in E} W'_{ij}(t) \mathbb{E}(\mu^{\pi'}_{ij}(t)| \bm{X}(t), \bm{X}'(t),\bm{\sigma}(t)) \nonumber \\
\geq \sum_{(i,j)\in E} W'_{ij}(t) \mathbb{E}(\mu^{\pi}_{ij}(t)| \bm{X}(t), \bm{X}'(t), \bm{\sigma}(t)) 
\end{eqnarray}
Taking expectation of both sides w.r.t. the random update process $\bm{X}'(t)$, conditioned on the true network state $\bm{X}(t)$ and the network configuration $\bm{\sigma}(t)$, we have 
\begin{eqnarray}
&& Cn^2c_{\max}/2+ \sum_{(i,j)\in E} W_{ij}(t) \mathbb{E}(\mu^{\pi'}_{ij}(t)| \bm{X}(t), \bm{\sigma}(t)) \nonumber \\
&\stackrel{(a)}\geq&\sum_{(i,j)\in E} \mathbb{E}W'_{ij}(t) \mathbb{E}(\mu^{\pi'}_{ij}(t)| \bm{X}(t), \bm{\sigma}(t)) \nonumber \\
&\stackrel{(b)}{\geq} &\sum_{(i,j)\in E} \mathbb{E}W'_{ij}(t) \mathbb{E}(\mu^{\pi}_{ij}(t)| \bm{X}(t), \bm{\sigma}(t))\nonumber \\
&\stackrel{(c)}{\geq}& \sum_{(i,j)\in E} W_{ij}(t) \mathbb{E}(\mu^{\pi}_{ij}(t)| \bm{X}(t), \bm{\sigma}(t)) - Cn^2c_{\max}/2\nonumber \\ \label{e1}
\end{eqnarray}
Here the inequality (a) and (c) follows from Lemma \eqref{update_lemma} and the fact that $|E|\leq n^2/2$ and $\mu_{ij}(t) \leq c_{\max}$. Inequality $(b)$ follows from Eqn. \eqref{w2}. Thus from Eqn. \eqref{eq:rand_drift} and \eqref{e1}, the expected conditional drift of the Lyapunov function under the policy $\pi'$, where the expectation is taken w.r.t. the random update and arrival process is upper-bounded as follows:
\begin{eqnarray*}
\Delta^{\pi'}(\bm{X}(t)|\bm{X}(t),\bm{\sigma}(t)) \leq B'- 2 \sum_{(i,j)\in E}  W_{ij}(t) \mathbb{E}[\mu^{\pi}_{ij}(t)\mid \bm{X}(t),\bm{\sigma}(t)]
\end{eqnarray*} 
with the constant $B'\equiv B|V|+2Cn^2c_{\max}$. Since the above inequality holds for any policy $\pi \in \Pi$, we can follow the exactly same steps in the proof of Theorem \eqref{pi_star_optimality} by replacing an arbitrary $\pi$ by $\pi^{\mathrm{RAND}}$ and showing that it has negative drift. 

\subsection{Proof of Proposition \ref{grid_network_capacity}} \label{grid_network_capacity_proof}
Like many proofs in this paper, this proof also has a converse and an achievability part. In the converse part, we obtain an upper bound of $\frac{2}{5}$ for the broadcast capacity $\lambda^*_{\mathrm{stat}}$ of the stationary grid network (i.e. when all links are ON w.p. $1$). In the achievability part, we show that this upper bound is tight.

\subsubsection*{Part (a): Proof of the Converse: $\lambda^*_{\mathrm{stat}}\leq \frac{2}{5}$} 

We have shown earlier that for the purpose of achieving capacity, it is sufficient to restrict our attention to stationary randomized policies only. Suppose a stationary randomized policy $\pi$ achieves a broadcast rate $\lambda$ and it activates edge $e\in E$ at every slot with probability $f_e$. 
 Then for the nodes $\texttt{a}$ and $\texttt{b}$ to receive distinct packets at rate $\lambda$, one requires
 \begin{eqnarray*}
 f_{\texttt{ra}}\geq \lambda, f_{\texttt{ab}}\geq \lambda 
 \end{eqnarray*}
 Applying the primary interference constraint at node $\mathsf{a}$, we then obtain 
 \begin{eqnarray*} \label{f_ae}
 f_{\texttt{ad}} \leq 1-2\lambda
 \end{eqnarray*}
Because of symmetry in the network topology, we also have 
\begin{eqnarray*}
f_{\texttt{cd}} \leq 1-2\lambda. 
\end{eqnarray*}
However, to achieve a broadcast capacity of $\lambda$, the total allocated rate towards node $\mathrm{d}$ must be atleast $\lambda$. Hence we have,

\begin{eqnarray*}
2(1-2\lambda) \geq \lambda
\end{eqnarray*}
i.e. 
\begin{eqnarray*}
\lambda \leq \frac{2}{5}.
\end{eqnarray*}
Since the above holds for any stationary randomized policy $\pi$, we conclude 
\begin{eqnarray} \label{cap_ub_ex}
\lambda^*_{\mathrm{stat}}\leq \frac{2}{5}  
\end{eqnarray}
$\blacksquare$

\begin{figure} [!ht] 
\centering
\begin{minipage}{\textwidth}
\begin{overpic}[width=0.25\textwidth]{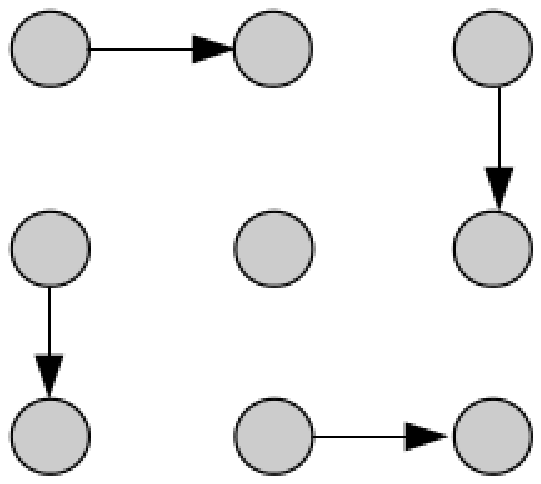}
\put(25,-5){Matching $M_1$}
\put(17,73){\texttt{r}}
\put(47,73){\texttt{a}}
\put(77,73){\texttt{b}}
\put(17,46){\texttt{c}}
\put(47,46){\texttt{d}}
\put(77,46){\texttt{e}}
\put(17,21){\texttt{f}}
\put(47,21){\texttt{g}}
\put(77,21){\texttt{h}}
\end{overpic}
%
 \begin{overpic}[width=0.25\textwidth]{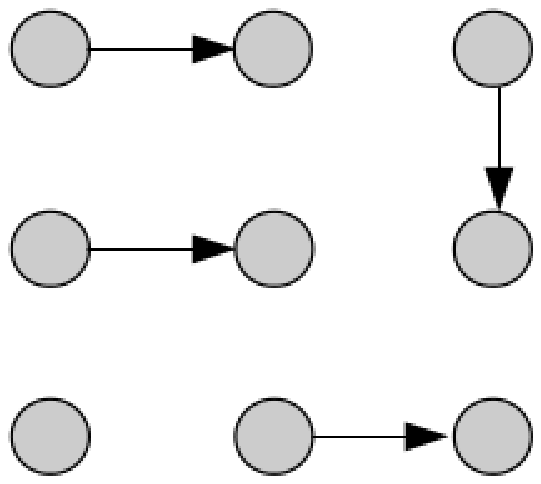}
  \put(25,-5){Matching $M_2$}
 \put(17,73){\texttt{r}}
\put(47,73){\texttt{a}}
\put(77,73){\texttt{b}}
\put(17,46){\texttt{c}}
\put(47,46){\texttt{d}}
\put(77,46){\texttt{e}}
\put(17,21){\texttt{f}}
\put(47,21){\texttt{g}}
\put(77,21){\texttt{h}}
  \end{overpic}
  \end{minipage}
\begin{minipage}{\textwidth}

 \begin{overpic}[width=0.25\textwidth]{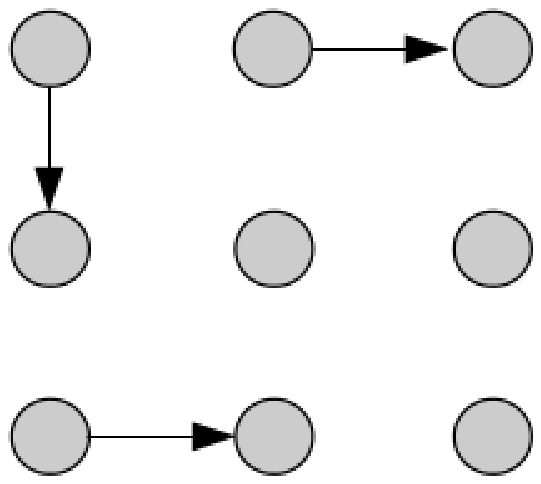}
  \put(25,-5){Matching $M_3$}
  \put(17,73){\texttt{r}}
\put(47,73){\texttt{a}}
\put(77,73){\texttt{b}}
\put(17,46){\texttt{c}}
\put(47,46){\texttt{d}}
\put(77,46){\texttt{e}}
\put(17,21){\texttt{f}}
\put(47,21){\texttt{g}}
\put(77,21){\texttt{h}}
  \end{overpic}
  \begin{overpic}[width=0.25\textwidth]{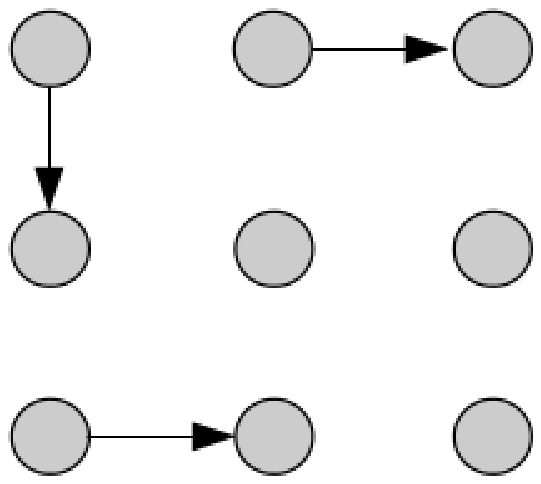}
  \put(25,-5){Matching $M_4$}
  \put(17,73){\texttt{r}}
\put(47,73){\texttt{a}}
\put(77,73){\texttt{b}}
\put(17,46){\texttt{c}}
\put(47,46){\texttt{d}}
\put(77,46){\texttt{e}}
\put(17,21){\texttt{f}}
\put(47,21){\texttt{g}}
\put(77,21){\texttt{h}}
  \end{overpic}
 
\end{minipage}

\begin{minipage}{\textwidth}
\begin{overpic}[width=0.25\textwidth]{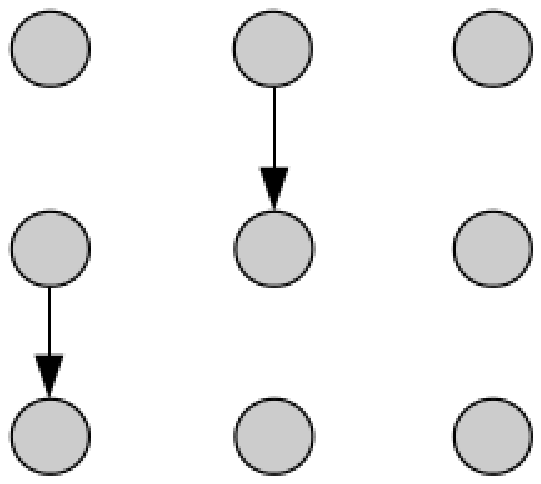}
  \put(25,-5){Matching $M_5$}
  \put(17,73){\texttt{r}}
\put(47,73){\texttt{a}}
\put(77,73){\texttt{b}}
\put(17,46){\texttt{c}}
\put(47,46){\texttt{d}}
\put(77,46){\texttt{e}}
\put(17,21){\texttt{f}}
\put(47,21){\texttt{g}}
\put(77,21){\texttt{h}}
  \end{overpic}
 \begin{overpic}[width=0.25\textwidth]{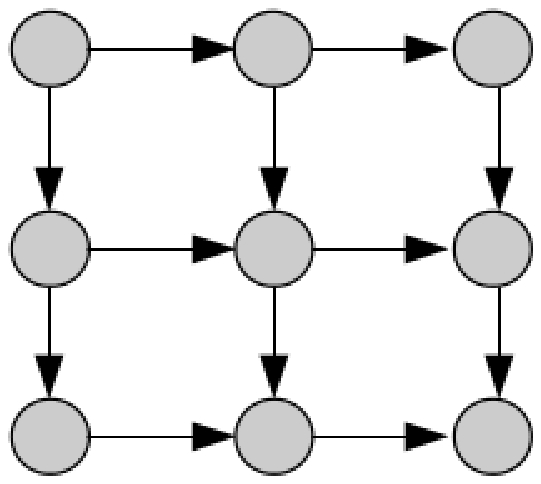}
  \put(7,-5){`Time averaged' Network}
  \put(17,73){\texttt{r}}
\put(47,73){\texttt{a}}
\put(77,73){\texttt{b}}
\put(17,46){\texttt{c}}
\put(47,46){\texttt{d}}
\put(77,46){\texttt{e}}
\put(17,21){\texttt{f}}
\put(47,21){\texttt{g}}
\put(77,21){\texttt{h}}
\put(60,15){$\frac{2}{5}$}
\put(30,15){$\frac{2}{5}$}
\put(30,40){$\frac{1}{5}$}
\put(30,80){$\frac{2}{5}$}
\put(51,60){$\frac{1}{5}$}
\put(82,60){$\frac{2}{5}$}
\put(60,80){$\frac{2}{5}$}
\put(10,60){$\frac{2}{5}$}
\put(10,33){$\frac{2}{5}$}
\put(51,33){$0$}
\put(73,33){$0$}
\put(61,42){$0$}
  \end{overpic}
  
  \vspace*{15pt}
 
\end{minipage}
\caption{Some feasible activations of the $3 \times 3$ grid network which are activated uniformly at random. The components corresponding to each edge in the resulting overall activation vector $\bm{\beta}$ is denoted by the numbers alongside the edges. }
  \label{matching_fig}
\end{figure}
\subsubsection*{Part (b): Proof of the Achievability: $\lambda^*_{\mathrm{stat}}\geq \frac{2}{5}$:} As usual, the achievability proof will be constructive. Consider the following five activations (matchings) $M_1,M_2,\ldots, M_5$ of the underlying graph as shown in Figure \ref{matching_fig}. Now consider a stationary policy $\pi^* \in \Pi^*$ that activates the matchings $M_1,\ldots, M_5$ at each slot uniformly at random with probability $\frac{1}{5}$ for each matching. The resulting `time-averaged' graph is also shown in Figure \ref{matching_fig}. Using Theorem \ref{cap_th}, it is clear that $\lambda^*_{\mathrm{stat}} \geq \frac{2}{5}$. Combining the above with the converse result in Eqn. \eqref{cap_ub_ex}, we conclude that, 
\begin{eqnarray*}
\lambda^*_{\mathrm{stat}}=\frac{2}{5}
\end{eqnarray*} 
$\blacksquare$
\end{document}